\documentclass{article}      

\usepackage[utf8]{inputenc}
\usepackage[english]{babel}

\usepackage{gensymb}
\usepackage{amssymb}

\usepackage{graphicx}
\usepackage{amssymb}
\usepackage{amsmath}
\usepackage{amsthm}

\usepackage[english]{babel}
\usepackage{epstopdf}

\usepackage{caption}
\usepackage{subcaption}

\newtheorem{theorem}{Theorem}[section]
\newtheorem{corollary}[theorem]{Corollary}
\newtheorem{lemma}[theorem]{Lemma}
\newtheorem{proposition}[theorem]{Proposition}

\usepackage{natbib}

\title{eSSVI Surface Calibration}  
\author{Leo Pasquazzi\footnote{email: leo.pasquazzi@unimib.it, ORCID https://orcid.org/0000-0002-2467-2667}\\
Dipartimento di Statistica e Metodi Quantitativi,\\
Università degli Studi di Milano - Bicocca, \\
Piazza dell'Ateneo Nuovo, 1 - 20126, Milano, Italy}      

\begin{document}             

\maketitle                   

\section*{Abstract}
In this work I test two calibration algorithms for the eSSVI volatility surface. The two algorithms are (i) the robust calibration algorithm proposed in \cite{Corbetta_Cohort_Laachir_Martini_2019} and (ii) the calibration algorithm in \cite{Mingone_2022}. For the latter I considered two types of weights in the objective function. I fitted 108 end-of-month SPXW options chains from the period 2012-2022. The option data come from FactSet. In addition to this empirical part, this paper contains also a theoretical contribution which is a sharpening of the Hendriks-Martini proposition about the existence of crossing points between two eSSVI slices.

\bigskip

\noindent \textbf{Keywords:} eSSVI volatility surface, Calibration, Arbitrage free interpolation

\bigskip


\section{Introduction}

At the Global Derivatives \& Risk Management conference in Madrid, \citet{Gatheral_2004} presented the SVI parametrization for the implied variance smile $w(k)$ at a fixed maturity. This parametrization reads
\[w(k)=a+b\left\{\rho(k-m)+\sqrt{(k-m)^{2}+\sigma^{2}}\right\}\]
where $k=\ln K/F$ is the natural logarithm of the ratio between the strike and the forward price of the underlying, and where $a\in\mathbb{R}$, $b\geq 0$, $|\rho|<1$, $m\in\mathbb{R}$ and $\sigma>0$ are parameters governing the shape and position of the smile. 

The following appealing features of the SVI parametrization are well known: 
\begin{enumerate}
\item[(i) ] the SVI smiles are asymptotically linear in $k$ as $|k|\rightarrow\infty$ and therefore consistent with Roger Lee's moment formula \citep{Lee_2004} and 
\item[(ii) ] the large maturity limit of the implied variance smile of a Heston model with correlation parameter $\rho$ is SVI with the same value of $\rho$ \citep{Gatheral_Jacquier_2011}. 
\end{enumerate}
However, it is also well known that the SVI parametrization, in its full generality, is not arbitrage free. For example, it is not difficult to show that $\inf_{k}w(k)=a+b\sigma\sqrt{1-\rho^{2}}$ and hence it should always be required that $a+b\sigma\sqrt{1-\rho^{2}}\geq 0$. However, the latter condition is not enough to rule out butterfly arbitrage as can be seen from a well known counterexample of Axel Vogt (see \citealt{Gatheral_Jacquier_2014}). Moreover, fitting the SVI parametrization to more than a single maturity date may produce slices that cross over each other which is equivalent to the existence of calendar spread arbitrage opportunities. To overcome these issues, \citet{Gatheral_Jacquier_2014} introduced the SSVI parametrization which is a global parametrization for the whole implied total variance \textit{surface} where the fixed maturity slices are restricted to a subfamily of the SVI parametrization (the first ''S'' in front of SSVI stands for ''\textit{surface}''). In the SSVI parametrization, the implied total variance surface is given by
\begin{equation}\label{SSVI_surface}
w_{t}(k)=\frac{\theta_{t}}{2}\left\{1+\rho\varphi(\theta_{t}) k+\sqrt{(\varphi(\theta_{t})k+\rho)^{2}+(1-\rho^{2})}\right\}
\end{equation}
where $\theta_{t}$ is the ATM implied total variance at maturity $t$, $|\rho|<1$ and $\varphi(\theta_{t})$ is a smooth function from $\mathbb{R}_{+}^{*}$ to $\mathbb{R}_{+}^{*}$ such that the limit $\lim_{t\rightarrow\infty}\theta_{t}\varphi(\theta_{t})$ exists in $\mathbb{R}$. According to Theorems 4.1 and 4.2 in \citet{Gatheral_Jacquier_2014}, the SSVI surface (\ref{SSVI_surface}) is free of calendar spread arbitrage if and only if
\begin{enumerate}
\item[C1) ] $\partial_{t} \theta_{t}\geq 0$ for all $t\geq 0$ 
\item[C2) ] and $0\leq\partial_{\theta}\left(\theta\varphi(\theta)\right)\leq \frac{1}{\rho^{2}}\left(1+\sqrt{1-\rho^{2}}\right)\varphi(\theta)$ for all $\theta>0$ (the upper bound is infinite when $\rho=0$)
\end{enumerate}
and it is free of butterfly arbitrage if for all $\theta>0$ the following two conditions are both satisfied:
\begin{enumerate}
\item[B1) ] $\theta\varphi(\theta)(1+|\rho|)<4$,
\item[B2) ] $\theta\varphi(\theta)^{2}(1+|\rho|)\leq 4$.
\end{enumerate}
The latter conditions are quite close to necessary as well. In fact, Lemma 4.2 in \citet{Gatheral_Jacquier_2014} says that absence of butterfly arbitrage holds only if $\theta\varphi(\theta)(1+|\rho|)\leq 4$ for all $\theta>0$ (which is only slightly weaker than B1) and that condition B2 becomes necessary when $\theta\varphi(\theta)(1+|\rho|)=4$.

In order to make the SSVI surface more flexible, \citet{Hendriks_Martini_2019} made the $\rho$-parameter maturity dependent as well and called the resulting implied total variance surface model \textit{eSSVI surface} (the "e" in front of SSVI stands for "\textit{extended}").  Proposition 3.1 in \citet{Hendriks_Martini_2019} provides necessary and sufficient conditions for the absence of calendar spread arbitrage between two time slices. In order to state these conditions, we will indicate the parameters of two slices with $\theta_{i}$, $\varphi_{i}=\varphi(\theta_{i})$ and $\rho_{i}$, where the subscript $i$ takes on the value $1$ or $2$ according to whether the closer ($i=1$) or farther ($i=2$) maturity date is referred to. Proposition 3.1 in \citet{Hendriks_Martini_2019} says that two time slices do not cross over each other only if 
\begin{itemize}
\item[N') ] $\frac{\theta_{2}}{\theta_{1}}\geq 1$ and $\left(\frac{\theta_{2}\varphi_{2}}{\theta_{1}\varphi_{1}}\rho_{2}-\rho_{1}\right)^{2}\leq \left(\frac{\theta_{2}\varphi_{2}}{\theta_{1}\varphi_{1}}-1\right)^{2}$
\end{itemize}
and that condition N along with condition S below is sufficient to rule out the existence of crossing points:
\begin{itemize}
\item[S) ] $\frac{\varphi_{2}}{\varphi_{1}}\leq 1$ or $(\frac{\theta_{2}\varphi_{2}}{\theta_{1}\varphi_{1}}\rho_{2}-\rho_{1})^{2}\leq(\frac{\theta_{2}}{\theta_{1}}-1)(\frac{\theta_{2}\varphi_{2}^{2}}{\theta_{1}\varphi_{1}^{2}}-1)$
\end{itemize}
However, condition N' and condition S are not jointly sufficient to rule out the existence of crossing point. In fact, as can be seen from Proposition \ref{corrected_proposition} in the appendix of this paper, 
\begin{itemize}
\item when $\frac{\theta_{2}}{\theta_{1}}=1$ there are no crossing points if and only if either (i) $\rho_{1}=\rho_{2}=0$ and $\varphi_{2}/\varphi_{1}\geq1$ or (ii) $\varphi_{2}/\varphi_{1}=\rho_{1}/\rho_{2}$ and $\rho_{1}^{2}\geq \rho_{2}^{2}$
\item and when $\frac{\theta_{2}}{\theta_{1}}\neq 1$ there are no crossing points if and only if condition S holds jointly with condition
\begin{itemize}
\item[N) ] $\frac{\theta_{2}}{\theta_{1}}> 1$ and $1-\frac{\theta_{2}\varphi_{2}}{\theta_{1}\varphi_{1}}\leq\frac{\theta_{2}\varphi_{2}}{\theta_{1}\varphi_{1}}\rho_{2}-\rho_{1}\leq \frac{\theta_{2}\varphi_{2}}{\theta_{1}\varphi_{1}}-1$.
\end{itemize}
\end{itemize}
Almost all of the proof of Proposition \ref{corrected_proposition} in the appendix of this paper is built on the main ideas of the proof of Proposition 3.1 in \citet{Hendriks_Martini_2019}. As far as I know the only novelty are the result about tangency points in Lemma \ref{lemma_second_sufficient_condition} and the two final Lemmas \ref{lemma_two_intersection_points} and \ref{lemma_one_intersection_point}.

\smallskip

The next section describes the calibration algorithms which I tested: the robust algorithm of \citet{Corbetta_Cohort_Laachir_Martini_2019} and the algorithm of \citet{Mingone_2022}. The test results are summarized in Section 3.

\section{Two calibration algorithms for the eSSVI surface}

\subsection{The robust algorithm of \citet{Corbetta_Cohort_Laachir_Martini_2019}}\label{algoritmo_robusto}

The robust calibration algorithm of \citet{Corbetta_Cohort_Laachir_Martini_2019} gives rise to eSSVI surfaces which satisfy conditions B1, B2, N and the first inequality of condition S. The key ingredient of this algorithm is a reparametrization of the SSVI slices which forces them to go through the data point $(k^{*}, \theta^{*})$ which is closest to the ATM forward implied total variance. To enforce this restriction the parameter $\theta$ is expressed in terms of the parameters $\rho$, $\varphi$ and the data-driven pair $(k^{*}, \theta^{*})$, i.e. $\theta$ is taken to be 
\[\theta=\theta^{*}-\rho\theta\varphi k^{*}:=\theta^{*}-\rho\psi k^{*}\]
which is a first order approximation to the solution of the equation $w(k^{*})=\theta^{*}$, where $w(k)$ is defined as in (\ref{SSVI_surface}). 

What are the allowed values for the new parameter $\psi:=\theta\varphi$? Of course, the non-negativity constraints on $\theta^{*}$, $\theta$ and $\varphi(\theta)=\varphi$ translate immediately to condition
\begin{itemize}
\item[R0) ] 
\begin{itemize}
\item[i) ] $0\leq \psi$ if $\rho k^{*}\leq0$
\item[ii) ] $0\leq \psi\leq \theta^{*}/(\rho k^{*})$ if $\rho k^{*}>0$
\end{itemize}
\end{itemize}
and condition B1 translates to 
\begin{itemize}
\item[R1) ] $\psi<\frac{4}{1+|\rho|}$
\end{itemize}
Moreover, using $\theta=\theta^{*}-\rho\psi k^{*}$ it is easily seen that condition B2 translates to
\begin{itemize}
\item[R2) ] $\psi^{-}\leq\psi\leq\psi^{+}$ where
\[\psi^{\pm}:=\frac{-2\rho k^{*}}{1+|\rho|}\pm\sqrt{\frac{4\rho^{2}(k^{*})^{2}}{(1+|\rho|)^{2}}+\frac{4\theta^{*}}{1+|\rho|}}.\]
\end{itemize}

Next, consider the conditions N and S for preventing calendar spread arbitrage. As in Section 1, the subscript $i=1$ will always refer to the closer maturity date. Keeping this mind it is not difficult to verify that condition N translates to
\begin{itemize}
\item[R3) ] $\theta^{*}_{2}-\rho_{2}\psi_{2}k^{*}_{2}\geq \theta^{*}_{1}-\rho_{1}\psi_{1}k^{*}_{1}$ and $\psi_{2}\geq \psi_{1}\max\left\{\frac{1+\rho_{1}}{1+\rho_{2}}, \frac{1-\rho_{1}}{1-\rho_{2}}\right\}$
\end{itemize}
and that the first inequality in condition S becomes
\begin{itemize}
\item[R4) ] $\psi_{2}(\theta^{*}_{1}-\rho_{1}\psi_{1}k^{*}_{1})\leq\psi_{1}(\theta^{*}_{2}-\rho_{2}\psi_{2}k^{*}_{2})$.
\end{itemize}
The algorithm suggested in \citet{Corbetta_Cohort_Laachir_Martini_2019} neglects the second inequality in S, which could hold when the first one fails and could therefore allow for more flexibility. 

Now, in order to describe the algorithm, we first observe that it is a sequential algorithm which starts from the closest maturity date and works forward in time finding optimal values of the parameters $k^{*}_{i}$, $\theta^{*}_{i}$, $\psi_{i}$ and $\rho_{i}$ for each each maturity date $i=1,2,\dots, n$ given the optimal values for the previous maturity date. It is based on the observation that given the parameter values for maturity date $i-1$, and given an arbitrary value of $\rho_{i}$, one can use the constraints R0 - R4 to find an interval of admissible values for $\psi_{i}$. The bounds of this interval depend on the sign of $\rho_{i}k^{*}_{i}$ and in order to write them down it is convenient to define 
\begin{gather*}
B_{1i}=\frac{4}{1+|\rho_{i}|} \quad\quad B_{2i}^{\pm}:=\psi^{\pm}_{i}:=\frac{-2\rho k_{i}^{*}}{1+|\rho_{i}|}\pm\sqrt{\frac{4\rho_{i}^{2}(k_{i}^{*})^{2}}{(1+|\rho_{i}|)^{2}}+\frac{4\theta^{*}_{i}}{1+|\rho_{i}|}},\\
B_{3i}:=\frac{\theta^{*}_{i}-\theta^{*}_{i-1}+\rho_{i-1}\psi_{i-1}k^{*}_{i-1}}{\rho_{i}k^{*}_{i}}\quad\quad B_{4i}:=\psi_{i-1}\max\left\{\frac{1+\rho_{i-1}}{1+\rho_{i}}, \frac{1-\rho_{i-1}}{1-\rho_{i}}\right\},\\
B_{5i}:=\frac{\psi_{i-1}\theta^{*}_{i}}{\theta^{*}_{i-1}-\psi_{i-1}(\rho_{i-1}k^{*}_{i-1}-\rho_{i}k^{*}_{i})}.
\end{gather*}
for $i=1,2,\dots, n$, where $\theta^{*}_{0}$, $k^{*}_{0}$, $\rho_{0}$ and $\psi_{0}$ are all defined to be zero. Now, it is not difficult to see that the constraints R0 - R4 are equivalent to the following bounds for $\psi_{i}$ for $i=1,2,\dots, n$:
\begin{itemize}
\item If $\rho_{i}k^{*}_{i}>0$, we get the lower bound $L_{i}:=\max\left\{B_{2i}^{-}, B_{4i}\right\}$ and the upper bound $U_{i}:=\min\left\{B_{1i}, B_{2i}^{+}, B_{3i}, B_{5i}\right\}$.
\item If $\rho_{i}k^{*}_{i}=0$, we get the lower bound $L_{i}:=B_{4i}$ and the upper bound $U_{i}:=\min\left\{B_{1i}, B_{2i}^{+}, B_{5i} \right\}$; in this case we must also check whether $\theta^{*}_{i}\geq \theta^{*}_{i-1}-\rho_{i-1}\psi_{i-1}k^{*}_{i-1}$, in order to make sure that the first inequality in condition R3 be satisfied.
\item If $\rho_{i}k^{*}_{i}<0$, we get the lower bound $L_{i}:=\max\left\{B_{2i}^{-}, B_{3i}, B_{4i}\right\}$ and the upper bound 
\[U_{i}:=\begin{cases}
\min\left\{B_{1i}, B_{2i}^{+}, B_{5i}\right\} & \text{ if }B_{5i}>0\\
\min\left\{B_{1i}, B_{2i}^{+}\right\} & \text{ if }B_{5i}\leq 0.
\end{cases}\]
\end{itemize}
Note that the bounds $L_{i}$ and $U_{i}$ depend on $k^{*}_{i-1}$, $k^{*}_{i}$, $\theta^{*}_{i}$, $\theta^{*}_{i-1}$, $\psi_{i-1}$, $\rho_{i-1}$ and $\rho_{i}$. In order to compute these bounds we must therefore know the optimal values of $\psi_{i-1}$ and $\rho_{i-1}$ which refer to the previous maturity date, and we must also guess a value for the parameter $\rho_{i}$ which refers to the current maturity date $i$. As we will see in a moment, the robust algorithm of \citet{Corbetta_Cohort_Laachir_Martini_2019} deals with this problem by calibrating the SSVI slices sequentially from the closest to farthest maturity date. As pointed out in \citet{Corbetta_Cohort_Laachir_Martini_2019}, one could actually adapt the algorithm and change the order in which the slices are calibrated, but usually there are good reasons not to do so.

Of course, if $L_{i}>U_{i}$, there will not exist any value of $\psi_{i}$ such that $(k^{*}_{i}, \theta^{*}_{i}, \rho_{i}, \psi_{i})$ gives rise to a SSVI slice which is free of butterfly arbitrage and/or calendar spread arbitrage with respect to the optimal slice for maturity $i-1$. For the given value of $\rho_{i}$ it is therefore impossible to avoid arbitrage. In this case the objective function to be minimized should be set equal to infinity in order to force the algorithm to consider other $\rho_{i}$ values for which $L_{i}\leq U_{i}$. For some maturity date $i$ it might also happen that $\theta^{*}_{i}<\theta^{*}_{i-1}-\rho_{i-1}\psi_{i-1}k^{*}_{i-1}$, in which case R3 cannot be satisfied with $\rho_{i}=0$. If $k^{*}_{i}\neq 0$, R3 can however be satisfied for every non zero value of $\rho_{i}$. Of course, nothing can be done in the remote case (which never occurred with the data I used) where $k^{*}_{i-1}=k^{*}_{i}=0$ and $\theta^{*}_{i}<\theta^{*}_{i-1}-\rho_{i-1}\psi_{i-1}k^{*}_{i-1}$: in this case we must increase the value of $\theta^{*}_{i}$ and make it equal $\theta^{*}_{i-1}-\rho_{i-1}\psi_{i-1}k^{*}_{i-1}$ or perhaps a little larger. Note that in the first case the algorithm might produce a slice which is affected by calendar spread arbitrage (see the case $\Theta=1$ in Proposition \ref{corrected_proposition} in the appendix). Keeping in mind these issues, we can now describe the algorithm:
\begin{itemize}
\item[0) ] Choose an objective function to minimize. In this work I used the sum of the absolute values of the differences between the observed option mid prices and their theoretical counterparts which would obtain if $w_{i}(k)$ was the Black and Scholes implied variance. Set $i=1$, i.e. start from the closest maturity date which is identified by the subscript $i=1$. 
\item[1) ] Set $\zeta=0$ (the role of $\zeta$ will become clear in a moment) and choose $r$ values for $\rho_{i}\in(-1,1)$ spaced apart at equal distances. These values will be denoted by $\rho_{i,j}$, $j=1,2,\dots r$. For each value of $\rho_{i,j}$ compute the corresponding bounds $L_{i,j}:=L_{i}$ and $U_{i,j}:=U_{i}$. In my empirical investigation I set $r=100$. At first sight this may seem a large value, but this choice is motivated by the fact that with smaller values of $r$ I obtained $L_{i,j}>U_{i,j}$ for every $j$ in some option chains.
\item[2) ] For each value of $\rho_{i,j}$, search the optimal value of $\psi_{i}$ within the corrseponding interval $(L_{i,j}, U_{i,j})$. The optimal value of $\psi_{i}$ corresponding to $\rho_{i,j}$ will be denoted $\psi_{i,j}$. In my empirical investigation I used the Brent method as implemented by the \texttt{fminbound} function from the \texttt{scipy.optimize} library for the minimum search. I set the maximum number of objective function evaluations to $1000$ and set \texttt{xtol} to \texttt{1e-8}.
\item[3) ] Compare the minima of the objective function obtained for the considered $\rho_{i,j}$-values and pick out $\rho_{i}^{*}:=\rho_{i,j}$ and $\psi_{i}^{*}:=\psi_{i,j}$ which give rise to the smallest minimum.
\item[4) ] Increment $\zeta$ by one unit and choose $r$ values for $\rho_{i}$ in the interval $(\rho_{i}^{*}-1.2/r^{\zeta}, \rho_{i}^{*}+1.2/r^{\zeta})$. Denote these values by $\rho_{i,j}$, $j=1,2,\dots r$ and go back to step 2) until $2\times 1.2\times r^{-\zeta}$ is smaller than a small value $\varepsilon$. In this work I used $\varepsilon=10^{-5}$.
\item[5) ] Define the parameters of the optimal slice at maturity date $i$ by setting $\rho_{i}:=\rho_{i}^{*}$ and $\psi_{i}:=\psi_{i}^{*}$. If the current $i$ corresponds to the last maturity date stop. Otherwise, increment the index $i$ and go back to step 1.
\end{itemize}

Once the algorithm is done we know the parameters $(\theta_{i}, \varphi_{i}, \rho_{i})$ which identify an optimal SSVI slice for each maturity date for which we have observed option prices. In order to extend these slices to a continuous variance surface we can employ the interpolation scheme described in \citet{Corbetta_Cohort_Laachir_Martini_2019}. Note that the calibrated surface is free of arbitrage of any kind.


\subsection{The algorithm of \citet{Mingone_2022}}\label{algoritmo_globale}

The algorithm proposed in \citet{Mingone_2022} is an attempt to overcome the sequential feature of the robust algorithm of \citet{Corbetta_Cohort_Laachir_Martini_2019} which may cause a poor fit at large maturity dates. To overcome this problem, Mingone introduces a reparametrization of the eSSVI surface which leads to a rectangular parameter space for which the conditions B1, B2, N and the first inequality in condition S are all satisfied. Actually, Mingone proposes also a second rectangular reparametrization where conditions B1, N and the first inequality in condition S are still satisfied, but where only a weaker condition than B2 holds. The second reparametrization accounts for the necessary and sufficient no butterfly arbitrage conditions which are given in \citet{Martini_Mingone_2022}. However, it is more difficult to implement than the first one and for this reason I tested only the calibration algorithm which is based on the first one. In order to describe the latter, it is convenient to recast conditions B1, B2, N and the first inequality in condition S in terms of $\theta$, $\psi:=\theta\varphi$ and $\rho$. Doing so we get the conditions
\begin{itemize}
\item[G1) ] $\psi<\frac{4}{1+|\rho|}$,
\item[G2) ] $\psi^{2}\leq\frac{4\theta}{1+|\rho|}$,
\item[G3) ] $\theta_{2}\geq \theta_{1}$ and $\psi_{2}\geq \psi_{1}\max\left\{\frac{1+\rho_{1}}{1+\rho_{2}}, \frac{1-\rho_{1}}{1-\rho_{2}}\right\}$,
\item[G4) ] $\psi_{2}\leq \psi_{1}\theta_{2}/\theta_{1}$.
\end{itemize}
Now, in order to get a rectangular parameter space, Mingone introduces the auxiliary quantities
\begin{align*}
p_{i}&:=\max\left\{\frac{1+\rho_{2}}{1+\rho_{1}}, \frac{1-\rho_{2}}{1-\rho_{1}}\right\} &\text{ for }i>1,  \\
f_{i}&:=\min\left\{\frac{4}{1+|\rho_{i}|}, \sqrt{\frac{4\theta_{i}}{1+|\rho_{i}|}}\right\}&\text{ for }i\geq 1,\\
A_{\psi_{1}}&:=0 & \text{ for }i= 1,\\
A_{\psi_{i}}&:=\psi_{i-1}p_{i} & \text{ for }i> 1,\\
C_{\psi_{1}}&:=\min\left\{f_{1}, \frac{f_{2}}{p_{2}}, \frac{f_{3}}{p_{2}p_{3}}, \dots, \frac{f_{n}}{\prod_{j=2}^{n}p_{j}}\right\} & \text{ for }i= 1,\\
C_{\psi_{i}}&:=\min\left\{\frac{\psi_{i-1}}{\theta_{i-1}}\theta_{i}, f_{i}, \frac{f_{i+1}}{p_{i+1}}, \frac{f_{i+2}}{p_{i+1}p_{i+2}}, \dots, \frac{f_{n}}{\prod_{j=i+1}^{n}p_{j}}\right\} & \text{ for }i> 1.
\end{align*}
and substitutes the old parameters $\theta_{2}$, $\theta_{3}$, \dots, $\theta_{n}$ and $\psi_{1}$, $\psi_{2}$, \dots, $\psi_{n}$ with the new parameters
\[a_{i}:=\theta_{i}-\theta_{i-1}p_{i},\quad i>1,\]
and
\[c_{i}:=\frac{\psi_{i}-A_{\psi_{i}}}{C_{\psi_{i}}-A_{\psi_{i}}},\quad i\geq 1,\]
respectively. Then she shows that for every choice of
\[(\rho_{1}, \dots, \rho_{n}, \theta_{1}, a_{2}, \dots, a_{n}, c_{1}, \dots, c_{n})\in(-1,1)^{n}\times(0,\infty)^{n}\times(0,1)^{n}\]
the conditions G1 - G4, with strict inequalities in place of weak ones, are all satisfied.

As emphasized by Mingone, her algorithm aims to overcome the problems caused by the sequential feature of the robust algorithm of \citet{Corbetta_Cohort_Laachir_Martini_2019}. For this reason she considers a global objective function of the form
\[\sum_{i=1}^{n}\sum_{j}[C(K_{j},i)-\widehat{C}(K_{j},i)]^{2} w(k_{j},i),\]
where $C(K_{j},i)$ is the observed option mid price at maturity $i$ and strike $K_{j}$, $j=1,2,\dots, J_{i}$,  $\widehat{C}(K_{j},i)$ is the theoretical option price which is obtained by applying the Black and Scholes formula to the suitable variance from the SSVI slice for maturity $i$, and where the $w(k_{j},i)$'s are positive weights. For the latter Mingone suggests to use the inverses of the squared market Black and Scholes vegas in order to achieve calibration in implied volatilities at the first order. In this work I followed this suggestion. For the sake of fair comparison with the robust algorithm, I also tried to use constant weights. For minimizing the objective function I used the Trust Region Reflective algorithm as implemented by the \texttt{least\_squares} function from the \texttt{scipy.optimize} library. I set the maximum number of objective function evaluations to $500$ and used the default value \texttt{1e-8} for the termination conditions \texttt{ftol}, \texttt{xtol} and \texttt{gtol}. As initial parameter vector for the optimization algorithm I used the parameter vector corresponding to the optimal SSVI slices from the robust algorithm. This choice is also suggested by Mingone. 

\section{Results}

This section describes the results of my empirical investigation. The data I used refer to SPXW options and are provided by FactSet. In order to cover a possibly wide range of market situations, I considered the option chains from the last trading date of each month starting from June 2012 until July 2022 ($122$ option chains). As usual in calibration exercises, I excluded some options from the analysis. To identify the excluded options, I first computed an implied forward price for the underlying. In order to do so, I considered all put-call pairs from which I computed an implied dividend yield according to the formula
\[q_{imp}:=\frac{1}{t}\left[\ln(C-P+Ke^{-rt})-\ln S\right]\]
where $t$ is time to maturity, $C$ and $P$ are the call and put mid-prices, $K$ is the strike price, $S$ is the underlying price and $r$ is the risk-free rate. To determine the latter I interpolated the US treasury yield curve provided by FactSet. Once I got the implied dividend yields, I averaged them for each maturity date and used the average values in order to compute forward prices for the underlying. Then I discarded all forward ITM options and all options with bid-ask percentage spread larger than $5\%$. From the remaining options I computed forward ATM implied volatilities for each maturity date $t$. I did this only for maturity dates where I was left with at least one option with forward-to-strike ratio larger and smaller than $1$. In this case, I considered the Black and Scholes implied volatility of the option whose forward-to-strike ratio is closer to $1$ as forward ATM implied volatility. Otherwise, I excluded all options with the given maturity date from the subsequent analysis. 


Having applied the above filtering operations, I was left with zero options for the $14$ chains of  
2012-09-28, 2018-01-31, 2013-05-31, 2017-05-31, 2013-06-28, 2012-10-31, 2018-03-29, 2012-11-30, 2012-12-31, 2018-05-31, 2013-01-31, 2013-11-29, 2017-04-28, 2012-08-31. Thus, my analysis refers to the remaining $122-14=108$ option chains. The two graphs in Figure \ref{Figura_serie_storiche_opzioni_disponibili} show the number of available SPXW options and the number of available maturity dates in each option chain before and after the filtering operations. 

\smallskip

For evaluating the fit of the calibrated eSSVI surfaces I computed four measures: \begin{itemize}
\item[i) ] the ratio 
\[F_{1}:=\frac{N_{w}}{N_{tot}}=\frac{\begin{array}{c}\text{number of theoretical option prices $\widehat{C}(K_{j},i)$}\\
\text{which are between the observed bid and ask prices}
\end{array}}{\text{total number of options used for calibration}}\]
\item[ii) ] the average absolute pricing error 
\[F_{2}:=\frac{1}{N_{tot}}\sum_{i=1}^{n}\sum_{j}|C(K_{j},i)-\widehat{C}(K_{j},i)|\]
\item[iii) ] the average of the squared pricing errors 
\[F_{3}:=\frac{1}{N_{tot}}\sum_{i=1}^{n}\sum_{j}|C(K_{j},i)-\widehat{C}(K_{j},i)|^{2}\]
\item[iv) ] the weighted average of squared pricing errors with weights $w(k_{j},i)$ given by the inverses of the squared market Black and Scholes vegas
\[F_{4}:=\frac{1}{\sum_{i=1}^{n}\sum_{j}w(k_{j},i)}\sum_{i=1}^{n}\sum_{j}|C(K_{j},i)-\widehat{C}(K_{j},i)|^{2}w(k_{j},i).\]
\end{itemize}
Of course, none of these measures on its own is perfectly fair for comparing the algorithms. In fact, one might expect $F_{1}$ to be in favor of the global algorithm with the quadratic inverse vega weighted objective function, since bid-ask spreads and vegas are usually positively correlated (see panel (a) in Figure \ref{correlations}). However, as can be seen from panels (a) and (b) in Figure \ref{Figura_measures_fit}, the global algorithm with constant weights seems to achieve smaller $F_{1}$ values on average. Also, in the comparison between the two global algorithms and the robust algorithm one might expect $F_{1}$ to be more helpful for the latter since bid-ask spreads and time to maturity are usually positively correlated (see panel (b) in Figure \ref{correlations}) and the robust algorithm tends to have problems to fit the larger maturity dates in chains with many different maturity dates. As can be seen from panel (a) in Figure \ref{Figura_measures_fit}, evidence for the latter conjecture is ambiguous.

Consider now the $F_{2}$ measure of fit. Also this measure can be expected to favor the robust algorithm since this algorithm aims to minimize sequentially, one after another, the inner sums in the definition of $F_{2}$. Moreover, in the comparison between the two implementations of the global algorithm $F_{2}$ should be in favor of the implementation with constant weights. The results shown in panel (c) and panel (d) of Figure \ref{Figura_measures_fit} are slightly supportive for these conjectures. 

Finally, as for the $F_{3}$ and $F_{4}$ measures, the results are as expected. The global algorithm with constant weights performs best w.r.t. to the $F_{3}$ measure, and the global algorithm with quadratic inverse vega weights is the best one w.r.t. $F_{4}$.

To complete the picture I will spend some words about computational aspects. As already mentioned, in both implementations of the global algorithm I used the optimal parameters from the robust algorithm as initial values. With this choice of initial values the \texttt{least\_squares} function was quite efficient in finding minima for the objective function. In fact, for the squared inverse vega weighted objective function the minimization process terminated before $500$ objective function evaluations for $93$ of the $108$ option chains because either the \texttt{ftol}, \texttt{gtol} or \texttt{xtol} termination condition was satisfied. For the objective function with constant weights the \texttt{least\_squares} function was able to find a minimum before $500$ function evaluations in $85$ of the $108$ option chains. However, my implementation of the robust algorithm could be too slow for some purposes. In fact, on average over all $2023$ calibrated eSSVI slices, it took about $5523$ objective function evaluations for finding the optimal parameters of a single slice. Motivated by this fact, I tested also the less data-driven initial values suggested in \citet{Mingone_2022} which are much faster to compute (these initial values are given by $\rho_{i}=0$, $a_{i}=\max\{\theta^{*}_{i}-\theta^{*}_{i-1}, 0\}$ and $c_{i}=0.5$ for $i=1,2,\dots, n$, where the $\theta^{*}_{i}$'s are the closest to ATM forward implied variances as defined in Section \ref{algoritmo_robusto}). As can be seen from Figure \ref{Figure_test_initial_values}, the results of this test are not very encouraging: the less data-driven initial values often lead to local minima where the $F_{1}$ measure of fit is much larger than what can be achieved through the robust initial values.


\section{Appendix: Revised proof of Proposition 3.1 in \citet{Hendriks_Martini_2019}}

Consider two eSSVI slices which we shall denote by
\[w_{i}(k)=\frac{\theta_{i}}{2}\left\{1+\rho_{i}\varphi_{i}k+\sqrt{\varphi_{i}^{2}k^{2}+2\varphi_{i}\rho_{i}k +1}\right\}, \quad i=1,2.\]
As in the main text, assume that the subscript $i=1$ refers to the closer maturity date. Then there is absence of calendar spread arbitrage if and only if $w_{1}(k)\leq w_{2}(k)$ for all $k\in\mathbb{R}$. 

Note that
\begin{equation*}
\begin{split}
w_{i}'(k)&=\frac{1}{2} \theta_{i}  \varphi_{i}  \left(\frac{k \varphi_{i} +\rho_{i} }{\sqrt{\varphi_{i}^{2}k^{2}+2\varphi_{i}\rho_{i}k +1}}+\rho_{i} \right)\\
w_{i}''(k)&=\frac{\theta_{i}  \left(1-\rho_{i}^{2}\right) \varphi_{i}^{2}}{2 (\varphi_{i}^{2}k^{2}+2\varphi_{i}\rho_{i}k +1)^{3/2}}
\end{split}
\end{equation*}
so that $w_{i}''(k)>0$ for all $k\in\mathbb{R}$. Since $w_{i}'(k)=0$ if and only if $k=k_{i}^{*}:=-\frac{2\rho_{i}}{\varphi_{i}}$, we conclude that
\[\inf_{k}w_{i}(k)=w_{i}(k_{i}^{*})=\theta_{i}(1-\rho_{i}^{2}).\]
By combining this result with the fact that $w_{i}(0)=\theta_{i}$, we see that absence of calendar spread arbitrage implies
\begin{equation}\label{noCA_1}
\Theta:=\frac{\theta_{2}}{\theta_{1}}\geq \max\left\{1, \frac{1-\rho_{1}^{2}}{1-\rho_{2}^{2}}\right\}.
\end{equation}
Another necessary condition may be obtained by considering the asymptotes of the two slices. Since 
\[2w_{i}(k)\sim\begin{cases}
\theta_{i}\varphi_{i}(1+\rho_{i})k & \text{ if }k\rightarrow \infty,\\
\theta_{i}\varphi_{i}(1-\rho_{i})k & \text{ if }k\rightarrow -\infty,
\end{cases}\]
we conclude that absence of calendar spread arbitrage also implies
\begin{equation}\label{noCA_2}
\Theta\Phi:=\frac{\theta_{2}\varphi_{2}}{\theta_{1}\varphi_{1}}\geq\max\left\{\frac{1+\rho_{1}}{1+\rho_{2}},\frac{1-\rho_{1}}{1-\rho_{2}}\right\}
\end{equation}
The latter condition is satisfied if and only if
\begin{equation}\label{noCA_2a}
\Theta\Phi\geq 1\quad\text{ and }\quad(\Theta\Phi\rho_{2}-\rho_{1})^{2}\leq(\Theta\Phi-1)^{2}.
\end{equation}

Of course, in the argument leading to the necessary condition (\ref{noCA_1}) we are tacitly assuming that $\varphi_{1}$, $\varphi_{2}$ and $\theta_{1}$ are all strictly positive and in this case it follows from (\ref{noCA_1}) that $\theta_{2}\geq\theta_{1}$, i.e. that $\Theta\geq 1$. Note that if $\varphi_{1}=0$ and/or $\theta_{1}=0$, then $w_{1}(k)=\theta_{1}$ for all $k\in\mathbb{R}$, and in this case we have absence of calendar spread arbitrage if and only if $\theta_{2}\geq\theta_{1}$ or $\theta_{2}(1-\rho_{2}^{2})\geq\theta_{1}$ according to whether $\varphi_{2}$ is also zero or not. On the other hand, if $\varphi_{2}=0$, then we have $w_{2}(k)=\theta_{2}$ for all $k\in\mathbb{R}$, and in this case it follows from the asymptotic behavior of $w_{1}(k)$ that we have absence of calendar spread arbitrage if and only if $\varphi_{1}=0$ and $\theta_{1}\leq\theta_{2}$. \textbf{In what follows we rule out these trivial cases by assuming that $\Phi:=\varphi_{2}/\varphi_{1}$ and $\Theta:=\theta_{2}/\theta_{1}$ are well defined (i.e. that their denominators are strictly positive) and that $\Phi>0$ and $\Theta\geq 1$.}

\begin{lemma}\label{lemma_necessary_conditions}
If $\theta_{1}$, $\varphi_{1}$ and $\varphi_{2}$ are all strictly positive, then there is absence of calendar spread arbitrage only if conditions (\ref{noCA_1}) and (\ref{noCA_2}) are both satisfied.
\end{lemma}
\smallskip

Now it arises the question whether the conditions (\ref{noCA_1}) and (\ref{noCA_2}) are sufficient as well. To answer this question we look for conditions under which the graphs of $w_{1}(k)$ and $w_{2}(k)$ have at least one point in common. I will proceed as in \citet{Hendriks_Martini_2019}, but I will try to make some steps more explicit. So let $x:=\varphi_{1}k$,
\[\alpha:=\alpha(x)=\Theta-1+(\Theta \Phi\rho_{2}-\rho_{1})x,\quad z_{1}:=z_{1}(x)=\sqrt{x^2+2 \rho_{1}x+1},\]
\[z_{2}:=z_{2}(x)=\sqrt{\Phi^{2}x^2+2 \rho_{2}\Phi x+1}\]
and note that the two eSSVI slices do have points in common if and only if the equation
\[\alpha(x)+\Theta z_{2}(x)=z_{1}(x)\]
has real solutions. Squaring twice yields the quartic polynomial
\[P(x):=4\alpha^{2}\Theta^{2}z_{2}^{2}-(z_{1}^{2}-\alpha^{2}-\Theta^{2}z_{2}^{2})^{2}\] 
where we have omitted the independent variable $x$ on the RHS. Note that every root of $P(x)$ must satisfy one (and only one) of the following conditions:
\begin{equation}\label{tipi_di_radice}
\begin{split}
\text{a) }& 2\alpha\Theta z_{2}=-(z_{1}^{2}-\alpha^{2}-\Theta^{2}z_{2}^{2})\quad\text{ and }\quad\alpha-\Theta z_{2}=\pm z_{1},\\
\text{b) }& 2\alpha\Theta z_{2}=z_{1}^{2}-\alpha^{2}-\Theta^{2}z_{2}^{2}\quad\text{ and }\quad\alpha+\Theta z_{2}=-z_{1},\\
\text{c) }& 2\alpha\Theta z_{2}=z_{1}^{2}-\alpha^{2}-\Theta^{2}z_{2}^{2}\quad\text{ and }\quad\alpha+\Theta z_{2}=z_{1}.
\end{split}
\end{equation}
Of course, a root of $P(x)$ is a point where the two slices intersect if and only if it satisfies condition c). To explore the existence of such roots we first observe that $P(x)=x^{2}Q(x)$, where 
\begin{equation}\label{Q_polynomial}
\begin{split}
Q(x):=& \left[\left(\Theta  \Phi \rho _2-\rho _1\right)^{2}-\left(\Theta\Phi-1\right)^2\right]\left[\left(\Theta\Phi+1\right)^2-\left(\Theta \Phi \rho _2-\rho _1\right)^{2}\right]x^{2}\\
&+ 4 \Theta \left[\rho _1 \left(-\Theta ^2 \Phi ^2+(\Theta -2) \Theta  \rho _2^2 \Phi ^2+2 \Theta  \Phi ^2-1\right)+\right.\\
&\quad\quad\quad\quad\left.+\rho _2 \Phi  \left(\Theta ^2 \rho _2^2 \Phi ^2-\Theta ^2 \Phi ^2+2 \Theta -1\right)+(1-2 \Theta ) \rho _2 \rho _1^2 \Phi +\rho _1^3\right]x\\
&+4 (\Theta -1) \Theta  \left(\Theta \Phi ^2 \rho _2^2-\Theta  \Phi ^2-\rho _1^2+1\right).
\end{split}
\end{equation}
Note that $x=0$ (which is a root of $P(x)$) is an intersection point if and only if $\theta_{1}=\theta_{2}$, i.e. if and only if $\Theta=1$ (in fact, $w_{1}(0)=\theta_{1}=\theta_{2}=w_{2}(0)$ if and only if $\Theta=1$). Assuming that this is the case, we will now find conditions under which the two slices do cross over each other. To this aim we consider their derivatives. With $\theta_{1}=\theta_{2}=\theta$ (i.e. with $\Theta=1$) we obtain
\[w_{i}'(0)=\theta\varphi_{i}\rho_{i}\quad\text{ and }\quad w_{i}''(0)=\frac{1}{2}\theta\varphi_{i}^{2}(1-\rho_{i}^{2}).\]
To rule out the possibility that the two slices cross over in $x=0$, we must therefore impose
\begin{equation}\label{condizioni_theta_1}
w_{1}'(0)=w_{2}'(0)\quad \text{ and }\quad w_{1}''(0)\leq w_{2}''(0).
\end{equation}
If either one of these conditions fails, the two slices cross over in $x=0$. Since we are assuming that the $\theta_{i}$'s and $\varphi_{i}$'s are all strictly positive, the conditions (\ref{condizioni_theta_1}) can be jointly satisfied only if 
\begin{itemize}
\item either $\rho_{1}=\rho_{2}=0$ and $\varphi_{2}\geq \varphi_{1}$, in which case it is easy to verify that $w_{2}(k)\geq w_{1}(k)$ for all $k\in\mathbb{R}$;
\item or $\Phi=\rho_{1}/\rho_{2}$ and $\rho_{1}^{2}\geq\rho_{2}^{2}$, in which case the constant term and the coefficient of $x$ in the polynomial $Q(x)$ do both vanish, and hence the two slices have no intersection points other than $x=0$.
\end{itemize}
These considerations prove the following lemma:
\begin{lemma}\label{lemma_caso_THETA1}
Assume that $\Phi$ and $\Theta$ are well defined and that $\Phi>0$. If $\Theta=1$, there is no calendar spread arbitrage if and only if either (i) $\rho_{1}=\rho_{2}=0$ and $\Phi\geq 1$ or (ii) $\Phi=\rho_{1}/\rho_{2}$ and $\rho_{1}^{2}\geq\rho_{2}^{2}$.
\end{lemma}
Note that conditions (\ref{noCA_1}) and (\ref{noCA_2}) do not imply condition (i) or (ii) of the previous lemma (take for example $\Phi=1.2$, $\rho_{1}=0.9$ and $\rho_{2}=0.81$) and the former are therefore not strong enough to rule out calendar spread arbitrage even if we restrict to the case where $\Theta=1$. 

\smallskip

Consider now what happens when $\Theta>1$. In this case $w_{1}(0)=\theta_{1}<\theta_{2}=w_{2}(0)$ and $x=0$ is therefore not an intersection point. To investigate the existence of intersection points we analyze the polynomial $Q(x)$. We begin with the following lemma:

\begin{lemma}\label{lemma_existence_of_roots_Q}
Assume that $\Phi>0$ and $\Theta>1$. Then $Q(x)$ is of second degree if and only if 
\begin{equation}\label{conditions_for_existence_of_roots_bis}
(\Theta\Phi\rho_{2}-\rho_{1})^{2}\neq(\Theta\Phi-1)^{2}.
\end{equation}
and in this case its discriminant is given by
\begin{equation}\label{discriminante}
\begin{split}
D&:=16\Theta(\rho_{1}^{2}-\Theta^{2}\Phi^{2}\rho_{2}^{2}+\Theta^{2}\Phi^{2}-1)^{2}\left[(\Theta\Phi\rho_{2}-\rho_{1})^{2}-(\Theta-1)(\Theta\Phi^{2}-1)\right]
\end{split}
\end{equation}
\end{lemma}

\begin{proof}
The coefficient of $x^{2}$ in $Q(x)$ vanishes if and only if either
\[(\Theta\Phi\rho_{2}-\rho_{1})^{2}=(\Theta\Phi-1)^{2}\quad\text{ or }\quad (\Theta\Phi\rho_{2}-\rho_{1})^{2}=(\Theta\Phi+1)^{2}.\]
The second condition implies $\Theta\Phi<0$, and hence we conclude that $Q(x)$ is of second degree if and only if condition (\ref{conditions_for_existence_of_roots_bis}) holds. In this case the discriminant of $Q(x)$ can be written as in expression (\ref{discriminante}).
\end{proof} 

From the previous lemma we know that $Q(x)$ must have real roots if condition (\ref{conditions_for_existence_of_roots_bis}) holds jointly with
\begin{equation}\label{conditions_for_existence_of_roots}
\rho_{1}^{2}-\Theta^{2}\Phi^{2}\rho_{2}^{2}+\Theta^{2}\Phi^{2}-1=0 \quad\text{ or }\quad(\Theta-1)(\Theta\Phi^{2}-1)\leq(\Theta\Phi\rho_{2}-\rho_{1})^{2}
\end{equation}
Since we are only concerned with the case where the necessary conditions (\ref{noCA_1}) and (\ref{noCA_2a}) hold, we must consider the set of $(\rho_{1},\rho_{2})$-pairs where the equality in (\ref{conditions_for_existence_of_roots}) holds and the set $(\rho_{1},\rho_{2})$-pairs where
\[(\Theta-1)(\Theta\Phi^{2}-1)\leq(\Theta\Phi\rho_{2}-\rho_{1})^{2}\leq (\Theta\Phi-1)^{2}.\]
We denote these two sets by $H_{\Theta,\Phi}$ and $R_{\Theta,\Phi}$, respectively:
\begin{equation*}
\begin{split}
H_{\Theta, \Phi}&:=\{ (\rho_{1},\rho_{2})\in(-1,1)^{2}: \rho_{1}^{2}-\Theta^{2}\Phi^{2}\rho_{2}^{2}+\Theta^{2}\Phi^{2}-1=0\},\\
R_{\Theta, \Phi}&:=\{ (\rho_{1},\rho_{2})\in(-1,1)^{2}: (\Theta-1)(\Theta\Phi^{2}-1)\leq(\Theta\Phi\rho_{2}-\rho_{1})^{2}\leq (\Theta\Phi-1)^{2}\}.
\end{split}
\end{equation*}
It will be useful to visualize these two sets. Since we have already dealt with the case $\Theta=1$, and since we are assuming the necessary condition (\ref{noCA_2a}), we need to consider only the case where $\Theta>1$ and $\Theta\Phi\geq 1$. Figure \ref{Figura_insieme_ammissibile} shows all possible shapes of $H_{\Theta, \Phi}$ and $R_{\Theta, \Phi}$. The set $H_{\Theta, \Phi}$ is the graph of a hyperbola. It is always symmetric with respect to both axes of the $(\rho_{1},\rho_{2})$-plane and its prolongation always goes through the four vertices of the square $[-1,1]^{2}\subset\mathbb{R}^{2}$. Moreover, 
\begin{itemize}
\item if $\Theta\Phi>1$, $H_{\Theta, \Phi}$ does not intersect the $\rho_{1}$ axis and intersects the $\rho_{2}$ axis in $\rho_{2}=\pm \sqrt{\frac{\Theta^{2}\Phi^{2}-1}{\Theta^{2}\Phi^{2}}}$;
\item if $\Theta\Phi=1$, $H_{\Theta, \Phi}$ reduces to the straight lines $\rho_{1}=\pm \rho_{2}$.
\end{itemize}
Also for the set $R_{\Theta, \Phi}$ there are essentially only two possible shapes:
\begin{itemize}
\item If $\Theta\Phi^{2}-1\leq 0$ (since we are assuming $\Theta>1$, this implies $\Phi<1$), the set $R_{\Theta, \Phi}$ is given by the stripe
\[S:=\{ (\rho_{1},\rho_{2})\in(-1,1)^{2}: \Theta\Phi\rho_{2}-\Theta\Phi+1\leq\rho_{1}\leq\Theta\Phi\rho_{2}+\Theta\Phi-1\}.\]
The stripe reduces to the line $\rho_{1}=\rho_{2}$ if $\Theta\Phi=1$.
\item If $\Theta\Phi^{2}-1>0$ (since we are assuming $\Theta>1$, this implies $\Theta\Phi>1$), the set $R_{\Theta, \Phi}$ is the union of the two parallel and disjoint stripes
\[S_{1}:=\{ (\rho_{1},\rho_{2})\in(-1,1)^{2}: \Theta\Phi\rho_{2}+\sqrt{(\Theta-1)(\Theta\Phi^{2}-1)}\leq\rho_{1}\leq\Theta\Phi\rho_{2}+\Theta\Phi-1\}\]
and
\[S_{2}:=\{ (\rho_{1},\rho_{2})\in(-1,1)^{2}: \Theta\Phi\rho_{2}-\Theta\Phi+1\leq\rho_{1}\leq\Theta\Phi\rho_{2}-\sqrt{(\Theta-1)(\Theta\Phi^{2}-1)}\}.\]
The two stripes reduce to the lines $\rho_{1}=\Theta\rho_{2}\pm(\Theta-1)$ if $\Phi=1$.
\end{itemize}
From the description of the sets $H_{\Theta, \Phi}$ and $R_{\Theta, \Phi}$ we see immediately that $H_{\Theta, \Phi}\cap R_{\Theta, \Phi}=\{ (\rho_{1},\rho_{2})\in(-1,1)^{2}:\rho_{1}=\rho_{2}\}$ when $\Theta\Phi=1$. Next we prove that except for this special case the intersection is empty.

\begin{lemma}\label{lemma_insieme_iperbola}
If $\Theta\Phi\neq 1$, then it follows that $S\cap H_{\Theta, \Phi}=\emptyset$.
\end{lemma}

\begin{proof}
If $(\rho_{1},\rho_{2})\in H_{\Theta, \Phi}$ we must have
\[\Theta\Phi(1-\rho_{2}^{2})=1-\rho_{1}^{2}\quad\Rightarrow\quad \Theta\Phi=\sqrt{\frac{1-\rho_{1}^{2}}{1-\rho_{2}^{2}}}.\]
From the latter equality we obtain
\begin{equation*}
\begin{split}
(\Theta\Phi\rho_{2}-\rho_{1})^{2}-(\Theta\Phi-1)^{2}=2\sqrt{\frac{1-\rho_{1}^{2}}{1-\rho_{2}^{2}}}(1-\rho_{1}\rho_{2})-2(1-\rho_{1}^{2}).
\end{split}
\end{equation*}
The quantity on the RHS is positive because 
\[\frac{1-\rho_{1}^{2}}{1-\rho_{2}^{2}}(1-\rho_{1}\rho_{2})^{2}>(1-\rho_{1}^{2})^{2}\quad\Leftrightarrow\quad (\rho_{1}-\rho_{2})^{2}>0\]
and because we are assuming that $\Theta\Phi=\sqrt{\frac{1-\rho_{1}^{2}}{1-\rho_{2}^{2}}}\neq 1$. Hence we conclude that $(\rho_{1},\rho_{2})\notin S$, because otherwise we should have
\[(\Theta\Phi\rho_{2}-\rho_{1})^{2}-(\Theta\Phi-1)^{2}\leq 0,\]
\end{proof}

From now on we consider roots of $Q(x)$ as functions of $\rho_{1}$ and $\rho_{2}$. Any root of $Q(x)$ will be denoted by $x_{(\rho_{1},\rho_{2})}$. Of course the subset of the $(\rho_{1},\rho_{2})$-plane where such a root exists depends on $\Theta$ and $\Phi$. We denote this set with $Q_{\Theta,\Phi}$. Since we are only interested in $\rho_{i}$ values within the interval $(-1,1)$, we consider $Q_{\Theta,\Phi}$ as a subset of the open square $(-1,1)^{2}$. Note that $x_{(\rho_{1},\rho_{2})}$ must be a continuous function of $(\rho_{1},\rho_{2})\in Q_{\Theta,\Phi}$ and that $(\rho_{1},\rho_{2})\in Q_{\Theta,\Phi}$ only if condition (\ref{conditions_for_existence_of_roots}) holds.

As we have already seen, a root $x_{(\rho_{1},\rho_{2})}$ must be of one of the three types in (\ref{tipi_di_radice}) and only roots of type c) are intersection points. In order to determine which type applies, we will need the following functions:
\[(\rho_{1},\rho_{2})\mapsto\alpha(x_{(\rho_{1},\rho_{2})}):=\Theta-1+(\Theta \Phi\rho_{2}-\rho_{1})x_{(\rho_{1},\rho_{2})}\]
and
\begin{equation*}
\begin{split}
(\rho_{1},\rho_{2})\mapsto Z(x_{(\rho_{1},\rho_{2})})&:=z_{1}^{2}(x_{(\rho_{1},\rho_{2})})-\Theta^{2}z_{2}^{2}(x_{(\rho_{1},\rho_{2})})-\alpha^{2}(x_{(\rho_{1},\rho_{2})})\\
&=x_{(\rho_{1},\rho_{2})}^{2}+2\rho_{1}x_{(\rho_{1},\rho_{2})}+1+\\
&\quad-\Theta^{2}\left(\Phi^{2}x_{(\rho_{1},\rho_{2})}^{2}+2\Phi\rho_{2}x_{(\rho_{1},\rho_{2})}+1\right)+\\
&\quad-\left[\Theta-1+(\Theta\Phi\rho_{2}-\rho_{1})x_{(\rho_{1},\rho_{2})}\right]^{2}.
\end{split}
\end{equation*}
Note that these functions must all be continuous functions of $(\rho_{1},\rho_{2})\in Q_{\Theta,\Phi}$.

\begin{lemma}\label{lemma_degli_zeri}
Assume that $\Phi>0$ and $\Theta>1$. Then $\alpha(x_{(\rho_{1},\rho_{2})})Z(x_{(\rho_{1},\rho_{2})})=0$ only if $(\rho_{1},\rho_{2})\in H_{\Theta,\Phi}$.
\end{lemma}

\begin{proof}
The proof is essentially the same as the proof of Lemma A.2 in \citet{Hendriks_Martini_2019}. Since $x_{(\rho_{1},\rho_{2})}$ must satisfy one of the two equalities
\[2\alpha(x_{(\rho_{1},\rho_{2})})\Theta z_{2}(x_{(\rho_{1},\rho_{2})}) = \pm Z(x_{(\rho_{1},\rho_{2})})\]
(otherwise $x_{(\rho_{1},\rho_{2})}$ is not a root of $P(x)$ and hence neither a root of $Q(x)$) and since $z_{2}(x)>0$ for all $x\in\mathbb{R}$, we conclude that $\alpha(x_{(\rho_{1},\rho_{2})})Z(x_{(\rho_{1},\rho_{2})})=0$ if and only if $\alpha(x_{(\rho_{1},\rho_{2})})$ and $Z(x_{(\rho_{1},\rho_{2})})$ do both vanish. In this case we must have 
\begin{equation}\label{dim_eq2}
\begin{gathered}
\alpha(x_{(\rho_{1},\rho_{2})})=0\quad\Rightarrow\quad \Theta\left(1+\rho_{2}\Phi x_{(\rho_{1},\rho_{2})}\right)=1+\rho_{1}x_{(\rho_{1},\rho_{2})}\quad\Rightarrow\\
\Theta^{2}\left(1+2\Phi\rho_{2} x_{(\rho_{1},\rho_{2})} + \Phi^{2}\rho_{2}^{2} x_{(\rho_{1},\rho_{2})}^{2}\right)=1+2\rho_{1}x_{(\rho_{1},\rho_{2})}+\rho_{1}^{2}x_{(\rho_{1},\rho_{2})}^{2}
\end{gathered}
\end{equation}
and we must also have 
\begin{equation}\label{dim_eq1}
\begin{gathered}
\Theta^{2}z_{2}^{2}(x_{(\rho_{1},\rho_{2})})=z_{1}^{2}(x_{(\rho_{1},\rho_{2})})\quad\Rightarrow\\
\Rightarrow\quad\Theta^{2}\left(\Phi^{2}x_{(\rho_{1},\rho_{2})}^{2}+2\Phi\rho_{2}x_{(\rho_{1},\rho_{2})}+1\right)=x_{(\rho_{1},\rho_{2})}^{2}+2\rho_{1}x_{(\rho_{1},\rho_{2})}+1,
\end{gathered}
\end{equation}
Subtracting equation (\ref{dim_eq2}) from equation (\ref{dim_eq1}) yields
\[x_{(\rho_{1},\rho_{2})}^{2}\left(\Theta^{2}\Phi^{2}-\Theta^{2}\Phi^{2}\rho_{2}^{2}-1+\rho_{1}^{2}\right)=0\]
Since $x_{(\rho_{1},\rho_{2})}$ must be different from zero (otherwise we would have $\alpha(x_{(\rho_{1},\rho_{2})})=\Theta-1>0$ contrary to our assumption), this implies that $\Theta^{2}\Phi^{2}-\Theta^{2}\Phi^{2}\rho_{2}^{2}-1+\rho_{1}^{2}=0$ which is equivalent to $(\rho_{1},\rho_{2})\in H_{\Theta,\Phi}$.
\end{proof}

\begin{corollary}\label{corollario_segno_alphaZ}
Assume that $\Phi>0$, $\Theta>1$ and that $A$ is a connected subset of $Q_{\Theta,\Phi}$ which does not intersect the set
\[H_{\Theta, \Phi}:=\{ (\rho_{1},\rho_{2})\in(-1,1)^{2}: \Theta^{2}\Phi^{2}(1-\rho_{2}^{2})=(1-\rho_{1}^{2})\}.\]
Then it follows that function $(\rho_{1},\rho_{2})\mapsto\alpha(x_{(\rho_{1},\rho_{2})})Z(x_{(\rho_{1},\rho_{2})})$ does not change sign on $A$. 
\end{corollary}

The previous corollary will be useful to distinguish whether a given root $x_{(\rho_{1},\rho_{2})}$ is of type a) rather than of type b) or c). Once we know that it is not of type a), we will apply the next lemma in order to find out whether it is of type b) or c).

\begin{lemma}\label{lemma_tipo_radice_bc}
Let $A$ be a connected subset of $Q_{\Theta, \Phi}$ such that $\alpha(x_{(\rho_{1},\rho_{2})})Z(x_{(\rho_{1},\rho_{2})})>0$ for all $(\rho_{1},\rho_{2})\in A$. Then it follows that the function $(\rho_{1},\rho_{2})\mapsto\alpha(x_{(\rho_{1},\rho_{2})})+\Theta z_{2}(x_{(\rho_{1},\rho_{2})})$ does not change sign on $A$.
\end{lemma}

\begin{proof}
Under the assumptions of the lemma $x_{(\rho_{1},\rho_{2})}$ must be a root of either type b) or c) in (\ref{tipi_di_radice}). Hence we must have either
\[\alpha(x_{(\rho_{1},\rho_{2})})+\Theta z_{2}(x_{(\rho_{1},\rho_{2})})= z_{1}(x_{(\rho_{1},\rho_{2})})\]
or 
\[\alpha(x_{(\rho_{1},\rho_{2})})+\Theta z_{2}(x_{(\rho_{1},\rho_{2})})=-z_{1}(x_{(\rho_{1},\rho_{2})}).\]
The conclusion of the lemma follows now from the fact that $\alpha(x_{(\rho_{1},\rho_{2})})+\Theta z_{2}(x_{(\rho_{1},\rho_{2})})$ is continuous and that $z_{1}(x)>0$ for all $x\in\mathbb{R}$.
\end{proof}

Now we are finally ready to investigate about the existence of intersection points. We start from the special cases where $\Phi=1$ or $\Theta\Phi=1$.

\begin{lemma}\label{casi_particolari}
Assume that $\Theta>1$ and that either $\Phi=1$ or $\Theta\Phi=1$. Then $(\rho_{1},\rho_{2})\in R_{\Theta,\Phi}$ implies $w_{1}(k)< w_{2}(k)$ for all $k\in\mathbb{R}$.
\end{lemma}

\begin{proof}
If $\Theta>1$, we must have $w_{1}(0)<w_{2}(0)$ and thus there exist intersection points only if the polynomial $Q(x)$ has real roots different from $x=0$. Now, consider first the case $\Phi=1$. Since we are assuming that $(\rho_{1},\rho_{2})\in R_{\Theta,\Phi}$, it follows that $\rho_{1}=\Theta\rho_{2}\pm(\Theta-1)$ (see the description of the set $R_{\Theta,\Phi}$ for the special case where $\Phi=1$). However, it can be verified that in this case we must have
\[Q(x)=-4 \Theta ^2(\Theta -1)^2  \left(\rho _2\pm 1\right){}^2<0\]
which has no roots at all.

Next, consider the case $\Theta\Phi=1$. In this case we must have $\rho_{1}=\rho_{2}$ (see the description of the set $R_{\Theta,\Phi}$ for the special case where $\Theta\Phi=1$). Substituting $\rho_{1}=\rho_{2}=\rho$ and $\Phi=1/\Theta$ in the coefficients of $Q(x)$ shows that
\[Q(x)=4 (\Theta -1)^2 \left(1-\rho^{2}\right)>0\]
which has no roots at all.
\end{proof}

Next, we deal with the case where $\Phi$ is strictly smaller than $1$ and different from $1/\Theta$ (i.e. $\Theta\Phi\neq 1$). The inequality $\Theta\Phi\geq1$, which is necessary by condition (\ref{noCA_2a}), allows then only for values of $\Phi$ in the range $1/\Theta<\Phi<1$. Note that for $\Phi\leq 1$ the necessary condition (\ref{noCA_1}) is already implied by condition (\ref{noCA_2a}) and therefore we do not need to assume condition (\ref{noCA_1}) explicitly.

\begin{lemma}\label{caso_unico_insieme_connesso}
Assume that $\Theta>1$, $\Theta\Phi> 1$ and $\Theta\Phi^{2}\leq1$ (this implies $\Phi<1$). Then $(\rho_{1},\rho_{2})\in R_{\Theta,\Phi}$ implies $w_{1}(k)< w_{2}(k)$ for all $k\in\mathbb{R}$.
\end{lemma}

\begin{proof}
Once again, if $\Theta>1$ we must have $w_{1}(0)<w_{2}(0)$ and there exist points $k\in\mathbb{R}$ where $w_{1}(k)\geq w_{2}(k)$ if and only if intersection points exist, i.e. if and only if the polynomial $Q(x)$ has at least one real root which satisfies condition c) in (\ref{tipi_di_radice}). From Lemma \ref{lemma_existence_of_roots_Q} we know that $Q(x)$ must have roots if $\Theta>1$, $\Theta\Phi> 1$, $\Theta\Phi^{2}\leq1$ and if $(\rho_{1},\rho_{2})$ belongs to the interior of $R_{\Theta,\Phi}$. Since under the present conditions $R_{\Theta,\Phi}$ is connected and does not intersect $H_{\Theta,\Phi}$ (see Lemma \ref{lemma_insieme_iperbola}), we may apply Corollary \ref{corollario_segno_alphaZ} to check whether the roots in $R_{\Theta,\Phi}$ are of type a). This will be the case if there exists a single $(\rho_{1},\rho_{2})\in int(R_{\Theta,\Phi})$ such that $\alpha(x_{(\rho_{1},\rho_{2})}) Z(x_{(\rho_{1},\rho_{2})})< 0$. Under the present conditions the origin belongs to $int(R_{\Theta,\Phi})$. Hence we use $(\rho_{1},\rho_{2})=(0,0)$ as test point. Of course, $\alpha(x_{(0,0)})=\Theta-1>0$. Moreover, it is easy to check that 
\[x_{(0,0)}=2 \frac{\sqrt{\Theta(\Theta -1)\left(1-\Theta  \Phi ^2\right)}}{\Theta ^2 \Phi ^2-1}\]
so that 
\[Z(x_{(0,0)})=-\frac{2 (\Theta -1) \Theta  \left[\Theta^{2}\Phi^{2}-1+2(1-\Theta  \Phi ^2)\right]}{\Theta ^2 \Phi ^2-1}<0.\]
We conclude that $x_{(\rho_{1},\rho_{2})}$ must be of type a) whenever $(\rho_{1},\rho_{2})\in R_{\Theta,\Phi}$.
\end{proof}

In the previous lemma we have assumed that $\Theta\Phi^{2}\leq1$ which forces $\Phi<1$. To apply the same method of proof for the case where $\Theta\Phi^{2}>1$ we must however \textit{assume} that $\Phi< 1$.

\begin{lemma}\label{caso_due_insiemi}
Assume $\Theta>1$, $\Theta\Phi^{2}>1$ and $\Phi< 1$ (note that $\Theta>1$ and $\Theta\Phi^{2}>1$ implies $\Theta\Phi>1$). Then $(\rho_{1},\rho_{2})\in R_{\Theta,\Phi}$ implies $w_{1}(k)< w_{2}(k)$ for all $k\in\mathbb{R}$.
\end{lemma}

\begin{proof}
The proof is similar to the proof of the previous lemma. However, in the present case we must deal with the fact that the set $R_{\Theta, \Phi}$ is not connected but only the union of the two connected sets $S_{1}$ and $S_{2}$. In each one of these two sets we must therefore find a point $(\rho_{1},\rho_{2})$ such that $\alpha(x_{(\rho_{1},\rho_{2})})$ and $Z(x_{(\rho_{1},\rho_{2})})$ are of opposite sign. To locate these points, note that the $\rho_{2}$-axis intersects both sets and hence we choose the $(\rho_{1},\rho_{2})$-points  with $\rho_{1}=0$ and 
\[\rho_{2}=\rho_{2}^{\pm}:=\pm\frac{\sqrt{(\Theta-1)(\Theta\Phi^{2}-1)}}{\Theta\Phi}.\]
This choice is convenient because it makes the discriminant of $Q(x)$ vanish. According to the sign in $\rho_{2}^{\pm}$, it gives rise to the roots
\begin{equation}\label{roots_rho_pm}
x_{(0,\rho_{2}^{\pm})}=\pm\frac{2 \sqrt{(\Theta -1) \left(\Theta  \Phi ^2-1\right)}}{\Theta(1 -\Phi ^2)}
\end{equation}
which, regardless of the sign, yields
\begin{equation}\label{alpha_rho_pm}
\alpha\left(x_{(0,\rho_{2}^{\pm})}\right)=\frac{(\Theta-1)(\Theta\Phi^{2}+\Theta-2)}{\Theta(1-\Phi^{2})}
\end{equation}
and
\begin{equation}\label{Z_rho_pm}
Z\left(x_{(0,\rho_{2}^{\pm})}\right)=-\frac{2 (\Theta -1) \left(\Theta  \Phi ^2+\Theta -2\right)^2}{\Theta  \left(\Phi ^2-1\right)^2}.
\end{equation}
Note that $Z\left(x_{(0,\rho_{2}^{\pm})}\right)<0$ regardless of the value of $\Phi$ (provided that $\Phi\neq1$), but to make sure that $\alpha\left(x_{(0,\rho_{2}^{\pm})}\right)>0$ we need to assume $\Phi<1$.
 
\end{proof}

Lemma \ref{casi_particolari}, Lemma \ref{caso_unico_insieme_connesso} and Lemma \ref{caso_due_insiemi} show that the necessary condition (\ref{noCA_2a}) along with $\Theta>1$ and $\Phi\leq 1$ are jointly sufficient to rule out calendar spread arbitrage. The next lemma deals with the condition
\begin{equation}\label{second_sufficient_condition}
(\Theta-1)(\Theta\Phi^{2}-1)\geq(\Theta\Phi\rho_{2}-\rho_{1})^{2}
\end{equation}
which allows for values of $\Phi$ larger than $1$.

\begin{lemma}\label{lemma_second_sufficient_condition}
Assume that $\Theta>1$ and that condition (\ref{second_sufficient_condition}) holds (note that these conditions jointly imply the necessary condition (\ref{noCA_2a})). Then it follows that $w_{1}(k)\leq w_{2}(k)$ for all $k\in\mathbb{R}$. Under the assumptions of this lemma there exist tangency points (i.e. values of $k$ where $w_{1}(k)=w_{2}(k)$) if and only if $\Phi>1$ and condition (\ref{second_sufficient_condition}) holds with equality sign. In that case there must exist exactly one tangency point.
\end{lemma}

\begin{proof}
If $\Theta>1$ and condition (\ref{second_sufficient_condition}) holds, we must have $\Theta\Phi^{2}\geq 1$ and hence $\Theta\Phi>1$ (otherwise there would not exist any $(\rho_{1},\rho_{2})$-pair for which (\ref{second_sufficient_condition}) holds). Consider first what happens when $\Theta\Phi^{2}=1$. In this case we must have $\Phi<1$ and for $\Phi\leq1$ we have already proved that $w_{1}(k)<w_{2}(k)$ for all $k\in\mathbb{R}$. 

Consider now what happens when $\Phi>1$. Since we are assuming that $\Theta>1$ (and hence we must have $\Theta\Phi^{2}>1$), we must have $R_{\Theta,\Phi}=S_{1}\cup S_{2}$ and the assumed inequality (\ref{second_sufficient_condition}) is satisfied if and only if the $(\rho_{1},\rho_{2})$-pair belongs to the area between the two stripes $S_{1}$ and $S_{2}$ or to one of the inner boundaries of $S_{1}$ or $S_{2}$. We indicate this set of $(\rho_{1},\rho_{2})$-pairs with $S_{3}$. Note that $S_{3}$ must be a proper subset of $S$ since we are assuming that $\Phi>1$. Since $\Phi>1$ implies $\Theta\Phi>1$, we can apply Lemma \ref{lemma_insieme_iperbola} and conclude that $S_{3}\cap H_{\Theta,\Phi}\subset S\cap H_{\Theta,\Phi}=\emptyset$. Now it follows from Lemma \ref{lemma_existence_of_roots_Q} that $Q(x)$ has no real roots when $(\rho_{1},\rho_{2})$ belongs to the interior of $S_{3}$, i.e. if the inequality (\ref{second_sufficient_condition}) is strict. Since the two slices can intersect only if $Q(x)$ has real roots, we conclude that $w_{1}(k)< w_{2}(k)$ for all $k\in\mathbb{R}$ in this case. On the other hand, if the $(\rho_{1},\rho_{2})$-pair belongs to the boundary of $S_{3}$, then it must also belong to the inner boundary of one of the two stripes $S_{1}$ or $S_{2}$. In other words, there must be equality in (\ref{second_sufficient_condition}) which means that 
\[\rho_{2}=\rho_{2}^{\pm}(\rho_{1}):=\frac{1}{\Theta\Phi}\left(\rho_{1}\pm \sqrt{(\Theta-1)(\Theta\Phi^{2}-1)}\right)\]
and that the discriminant of $Q(x)$ must be zero (see Lemma \ref{lemma_existence_of_roots_Q}). Hence $Q(x)$ must have a root and this root must be unique. As usual we write $x_{(\rho_{1}, \rho_{2})}$ to indicate the root. Since we are assuming that $\Theta$ and $\Phi$ are both larger than $1$, we can apply Lemma \ref{lemma_degli_zeri} and conclude that $\alpha(x_{(\rho_{1}, \rho_{2})})Z(x_{(\rho_{1}, \rho_{2})})\neq 0$ when the $(\rho_{1}, \rho_{2})$-pair belongs to $R_{\Theta,\Phi}=S_{1}\cup S_{2}$ and hence that $\alpha(x_{(\rho_{1}, \rho_{2})})Z(x_{(\rho_{1}, \rho_{2})})\neq 0$ for all $(\rho_{1}, \rho_{2})$-pairs which belong to the boundary of $S_{3}$ where a root $x_{(\rho_{1},\rho_{2})}$ must exist and must be unique. Since $S_{1}$ and $S_{2}$ are two disjoint and connected sets, $\alpha(x_{(\rho_{1}, \rho_{2})})Z(x_{(\rho_{1}, \rho_{2})})$ does not change sign on each of these two sets. We will now show that the sign of $\alpha(x_{(\rho_{1}, \rho_{2})})Z(x_{(\rho_{1}, \rho_{2})})$ is positive on both sets. This can be done by proving that the sign of $\alpha(x_{(\rho_{1}, \rho_{2})})Z(x_{(\rho_{1}, \rho_{2})})$ is positive at a single point in each of the two sets (see Corollary \ref{corollario_segno_alphaZ}). As in the proof of Lemma (\ref{caso_due_insiemi}) we use the $(\rho_{1},\rho_{2})$-pairs with $\rho_{1}=0$ and $\rho_{2}=\rho_{2}^{\pm}:=\pm\frac{1}{\Theta\Phi}\sqrt{(\Theta-1)(\Theta\Phi^{2}-1)}$ as test points (note that these points also belong to the boundary of $S_{3}$). With this choice we still get the expressions in (\ref{roots_rho_pm}), (\ref{alpha_rho_pm}) and (\ref{Z_rho_pm}) for $x_{(0, \rho_{2}^{\pm})}$, $\alpha(x_{(0, \rho_{2}^{\pm})})$ and $Z(x_{(0, \rho_{2}^{\pm})})$. However, since we are now assuming that $\Phi>1$, we see that $\alpha(x_{(0,\rho_{2}^{\pm})})<0$ and not $\alpha(x_{(0, \rho_{2}^{\pm})})>0$ as in the proof of Lemma \ref{caso_due_insiemi} (of course, $Z(x_{(0, \rho_{2}^{\pm})})$ remains still negative). We conclude that the roots we are considering now must be either of type b) or c) in (\ref{tipi_di_radice}). Hence we must have $\alpha(x_{(\rho_{1}, \rho_{2})})+\Theta z_{2}(x_{(\rho_{1}, \rho_{2})})=\pm z_{1}(x_{(\rho_{1}, \rho_{2})})$. In order to prove that the roots correspond to intersection points, we first note that the mapping $(\rho_{1},\rho_{2})\mapsto \alpha(x_{(\rho_{1}, \rho_{2})})+\Theta z_{2}(x_{(\rho_{1}, \rho_{2})})$ does not change sign on each of the two sets $S_{1}$ and $S_{2}$ (use Lemma \ref{lemma_tipo_radice_bc}). However, as far as we know by now, the sign of $\alpha(x_{(\rho_{1}, \rho_{2})})+\Theta z_{2}(x_{(\rho_{1}, \rho_{2})})$ might be different according to whether $(\rho_{1},\rho_{2})$ belongs to $S_{1}$ or to $S_{2}$. Thus, if there exists a single point $(\rho_{1},\rho_{2})\in S_{i}$ such that the sign of $\alpha(x_{(\rho_{1}, \rho_{2})})+\Theta z_{2}(x_{(\rho_{1}, \rho_{2})})$ is positive, we can conclude that $x_{(\rho_{1}, \rho_{2})}$ is of type c) and hence that $w_{1}(x_{(\rho_{1}, \rho_{2})}/\varphi_{1})=w_{2}(x_{(\rho_{1}, \rho_{2})}/\varphi_{1})$ for every $(\rho_{1}, \rho_{2})\in S_{i}$ ($i=1,2$). Again, we use the two $(\rho_{1},\rho_{2})$-pairs with $\rho_{1}=0$ and $\rho_{2}=\rho_{2}^{\pm}:=\pm\frac{1}{\Theta\Phi}\sqrt{(\Theta-1)(\Theta\Phi^{2}-1)}$ as test points. For these points we get the same expressions of $x_{(0,\rho_{2}^{\pm})}$ and $\alpha(x_{(0,\rho_{2}^{\pm}})$ as in equations (\ref{roots_rho_pm}) and (\ref{alpha_rho_pm}), while for $z_{2}(x_{(0,\rho_{2}^{\pm})})$ we get the expression
\[z_{2}(x_{(0,\rho_{2}^{\pm})})=\frac{\Theta\Phi^{2}+\Theta-2}{\Theta(\Phi^{2}-1)}.\]
Hence we conclude that
\[\alpha(x_{(0,\rho_{2}^{\pm})})+\Theta z_{2}(x_{(0,\rho_{2}^{\pm})})=\frac{\Theta\Phi^{2}+\Theta-2}{\Theta(\Phi^{2}-1)}>0.\]
This argument shows that every root $x_{(\rho_{1}, \rho_{2})}$ with $(\rho_{1}, \rho_{2})\in S_{1}\cup S_{2}$ is an intersection point and that there must be a unique intersection point when $(\rho_{1}, \rho_{2})\in(S_{1}\cup S_{2})\cap S_{3}$. It is not difficult to show that in the latter case the intersection point must be a tangency point. In fact, if it was a crossing point, there should exist one further crossing point because under our present assumptions the left and right asymptotes of $w_{2}(k)$ are both steeper than those of $w_{1}(k)$ (recall that we are assuming $\Theta>1$ and $\Phi>1$: on $(S_{1}\cup S_{2})\cap S_{3}$ condition (\ref{noCA_2a}) must therefore hold with strict inequality sign). \end{proof}


As far as I know, the results in the next two lemmas are new.

\begin{lemma}\label{lemma_two_intersection_points}
If 
\begin{equation}\label{condizione_two_intersection_points}
\Theta>1, \quad \Phi>1,\quad\text{ and }\quad(\Theta-1)(\Theta\Phi^{2}-1)<(\Theta\Phi\rho_{2}-\rho_{1})^{2}<(\Theta\Phi-1)^{2},
\end{equation}
there must exist exactly two points where the slices $w_{1}(k)$ and $w_{2}(k)$ cross over each other.
\end{lemma}

\begin{proof}
From Lemma \ref{lemma_existence_of_roots_Q} and Lemma \ref{lemma_insieme_iperbola} we know that under condition (\ref{condizione_two_intersection_points}) there must exist two roots $x_{(\rho_{1},\rho_{2})}$. Moreover, from the proof of Lemma \ref{lemma_second_sufficient_condition} we know that both these roots must be intersection points. The chained inequality in (\ref{condizione_two_intersection_points}) says that both asymptotes of $w_{2}(k)$ are steeper than those of $w_{1}(k)$. Therefore only two cases can occur: either (i) both intersection points are tangency points, or (ii) both intersection points are crossing points. 
It is not difficult to see that case (i) is impossibile. In fact, if both intersection points were tangency points, any increase of $\Theta$ should lead to $w_{1}(k)<w_{2}(k)$ for all $k\in\mathbb{R}$. However, given fixed values of $\Phi$, $\rho_{1}$ and $\rho_{2}$, a small enough increase of $\Theta$ does not lead to a violation of condition (\ref{condizione_two_intersection_points}) which implies the existence of two intersection points.
\end{proof}

Now, it remains to see what happens when
\begin{equation}\label{condizione_one_intersection_point}
\Theta>1, \quad \Phi>1,\quad\text{ and }\quad (\Theta\Phi\rho_{2}-\rho_{1})^{2}=(\Theta\Phi-1)^{2}
\end{equation}
i.e. when the left or right asymptote of $w_{2}(k)$ is the same as the corresponding asymptote of $w_{1}(k)$.

\begin{lemma}\label{lemma_one_intersection_point}
Assume condition (\ref{condizione_one_intersection_point}) holds. Then there must exist exactly one point where the slices $w_{1}(k)$ and $w_{2}(k)$ cross over each other.
\end{lemma}

\begin{proof}
Define $B_{\Theta,\Phi}$ as the subset of the $(\rho_{1},\rho_{2})$-plane where the equality in condition (\ref{condizione_one_intersection_point}) holds. $B_{\Theta,\Phi}$ is then the boundary of $S$, i.e. the subset of the $(\rho_{1},\rho_{2})$-plane where 
\[\rho_{1}=\rho_{1}^{+}(\rho_{2}):=\Theta\Phi\rho_{2}+\Theta\Phi-1\quad\text{ and }\quad-1<\rho_{2}<\min\left\{\frac{2}{\Theta\Phi}-1,1\right\}\]
or
\[\rho_{1}=\rho_{1}^{-}(\rho_{2}):=\Theta\Phi\rho_{2}+1-\Theta\Phi\quad\text{ and }\quad \max\left\{1-\frac{2}{\Theta\Phi},-1\right\}<\rho_{2}<1.\]
On $B_{\Theta,\Phi}$ the polynomial $Q(x)$ reduces to
\begin{equation*}
Q^{\pm}(x)=-4 \Theta ^2 \left(\rho _2\pm 1\right) \Phi  \left[(\Theta -1)^2 \Phi\rho _2 \pm (\Theta -1) (\Theta  \Phi +\Phi -2)-2 x (\Phi -1) (\Theta  \Phi -1)\right]
\end{equation*}
and the only root of $Q^{\pm}(x)$ is given by
\[x_{(\rho_{1}^{\pm}(\rho_{2}), \rho_{2})}=\frac{(\Theta -1)^2 \Phi\rho _2 \pm (\Theta -1) (\Theta  \Phi +\Phi -2)}{2(\Phi -1) (\Theta  \Phi -1)}.\]
We will show that this root must be a crossing point. To this aim note that $B_{\Theta,\Phi}$ is the union of two disjoint connected sets which we denote with $B_{\Theta,\Phi}^{\pm}$. From Lemma \ref{lemma_insieme_iperbola} it follows that $B_{\Theta,\Phi}\cap H_{\Theta,\Phi}=\emptyset$. Hence we may apply Corollary \ref{corollario_segno_alphaZ} and conclude that $\alpha(x_{(\rho_{1}^{\pm},\rho_{2})}) Z(x_{(\rho_{1}^{\pm},\rho_{2})})$ does not change sign on each one of the two connected components of $B_{\Theta,\Phi}=B_{\Theta,\Phi}^{+}\cup B_{\Theta,\Phi}^{-}$. In order to show that $\alpha(x_{(\rho_{1}^{\pm},\rho_{2})})$ and $Z(x_{(\rho_{1}^{\pm},\rho_{2})})$ are of the same sign, it is therefore sufficient to find a single point $(\rho_{1},\rho_{2})$ in each of the two sets $B_{\Theta,\Phi}^{+}$ and $B_{\Theta,\Phi}^{-}$ for which $\alpha(x_{(\rho_{1}^{\pm},\rho_{2})})$ and $Z(x_{(\rho_{1}^{\pm},\rho_{2})})$ are of the same sign. As test points we choose the points $(\rho_{1},\rho_{2})=(\rho_{1}^{\pm}(\rho_{2}), \rho_{2})$ where $\rho_{1}^{\pm}(\rho_{2})=0$. It is easily seen that these points are $(0,\pm\rho_{*})$ where $\rho_{*}=\frac{1}{\Theta\Phi} - 1$. Substituting in the formula for the root we get
\[x_{(0, \pm\rho_{*})}=\pm\frac{2 \Theta ^2 \Phi -\Theta ^2-2 \Theta  \Phi +1}{2 \Theta  (\Phi -1) (\Theta  \Phi -1)}.\]
Regardless of the sign in $\rho_{2}=\rho_{*}$, this root yields
\[\alpha(x_{(0, \pm\rho_{*})})=-\frac{(\Theta-1)^2}{2 \Theta  (\Phi - 1)}\]
and 
\[Z(x_{(0, \pm\rho_{*})})=-\frac{(\Theta -1)^2 (\Theta ^2 \Phi +2 \Theta  \Phi ^2-4 \Theta  \Phi -\Phi +2)}{2 \Theta  (\Phi -1)^2 (\Theta  \Phi -1)}.\]
Of course $\alpha(x_{(0, \pm\rho_{*})})<0$. As for $Z(x_{(0, \pm\rho_{*})})$, its sign depends on the sign of
\[\Theta ^2 \Phi +2 \Theta  \Phi ^2-4 \Theta  \Phi -\Phi +2\]
which is positive whenever $\Theta>1$ and $\Phi>1$ (we omit the details of the proof of this assertion). Hence also $Z(x_{(0, \pm\rho_{*})})<0$ and thus we conclude that $x_{(0, \pm\rho_{*})}$ is a root of either type b) or c) in (\ref{tipi_di_radice}). To find out which type applies, we must determine the sign of $\alpha(x_{(0, \pm\rho_{*})})+\Theta z_{2}(x_{(0, \pm\rho_{*})})$. It is not difficult to verify that
\[z_{2}(x_{(0, \pm\rho_{*})})=\frac{\Theta ^2 \Phi +2 \Theta  \Phi ^2-4 \Theta  \Phi -\Phi +2}{2 \Theta  (\Phi -1) (\Theta  \Phi -1)}\]
and hence we get
\[\alpha(x_{(0, \pm\rho_{*})})+\Theta z_{2}(x_{(0, \pm\rho_{*})})=\frac{2 \Theta ^2 \Phi ^2-2 \Theta ^2 \Phi +\Theta ^2-2 \Theta  \Phi +1}{2 \Theta  (\Phi -1) (\Theta  \Phi -1)}.\]
The numerator in this expression can be written as
\[\Theta^{2}(\Phi-1)^{2}+(\Theta\Phi-1)^{2}\]
and therefore we must have $\alpha(x_{(0, \pm\rho_{*})})+\Theta z_{2}(x_{(0, \pm\rho_{*})})>0$. By Lemma \ref{lemma_tipo_radice_bc} we conclude that $x_{(\rho_{1}, \rho_{2})}$ must be an intersection point for every $(\rho_{1}, \rho_{2})\in B_{\Theta,\Phi}$. 

To complete the proof it remains to show that for every $(\rho_{1}, \rho_{2})\in B_{\Theta,\Phi}$ the corresponding root $x_{(\rho_{1}, \rho_{2})}$ is a crossing point. To this aim we apply the argument in the proof of Lemma \ref{lemma_two_intersection_points} once again: if $x_{(\rho_{1}, \rho_{2})}$ was a tangency point, then by increasing $\Theta$ a little bit we should have no intersection points at all. However, if we increase $\Theta$ a little bit while leaving $\Phi$, $\rho_{1}$ and $\rho_{2}$ unchanged, we pass from condition (\ref{condizione_one_intersection_point}) to condition (\ref{condizione_two_intersection_points}) which implies the existence of two crossing points.

\end{proof}

Combining the statements in Lemma \ref{lemma_necessary_conditions}, Lemma \ref{lemma_caso_THETA1}, Lemma \ref{casi_particolari}, Lemma \ref{caso_unico_insieme_connesso}, Lemma \ref{caso_due_insiemi}, Lemma \ref{lemma_second_sufficient_condition}, Lemma \ref{lemma_two_intersection_points} and Lemma \ref{lemma_one_intersection_point} yields a corrected and sharper version of Proposition 3.1 in \citet{Hendriks_Martini_2019}. The little corrections concern 
\begin{itemize}
\item[i) ] the special case where $\Theta=1$ and
\item[ii) ] the fact that Proposition 3.1 in \citet{Hendriks_Martini_2019} seems to imply that with $\Theta>1$, $(\Theta\Phi\rho_{2}-\rho_{1})^{2}\leq(\Theta\Phi-1)^{2}$ and $\Phi< 1$ it would be possible to have absence of calendar spread arbitrage even if $\Theta\Phi <1$ which goes against the necessary condition (\ref{noCA_2a}). However, the preprint version of the article contains a slightly different version of the proposition which is not subject to this problem but where the two necessary conditions are a little too strong due to strong inequality signs instead of weak ones (see Proposition 3.5 in \citet{Hendriks_Martini_2017}). 
\end{itemize}

The sharper (and corrected) statement of the Hendriks-Martini proposition is given below. To make it more concise, the necessary condition (\ref{noCA_2}) will be stated as in (\ref{condizione_necessaria_2}).

\begin{proposition}\label{corrected_proposition}
Assume that $\theta_{1}$ and $\varphi_{1}$ are both strictly positive and let $\Theta:=\theta_{2}/\theta_{1}$ and $\Phi:=\varphi_{2}/\varphi_{1}$. Then, there is absence of calendar spread arbitrage (i.e. $w_{1}(k)\leq w_{2}(k)$ for all $k\in\mathbb{R}$) only if $\Theta\geq 1$ and
\begin{equation}\label{condizione_necessaria_2} 
1-\Theta\Phi\leq\Theta\Phi\rho_{2}-\rho_{1}\leq\Theta\Phi-1
\end{equation}

Moreover, 
\begin{itemize}
\item when $\Theta=1$ there is absence of calendar spread arbitrage if and only if either (i) $\rho_{1}=\rho_{2}=0$ and $\Phi\geq 1$ or (ii) $\Phi=\rho_{1}/\rho_{2}$ and $\rho_{1}^{2}\geq\rho_{2}^{2}$;
\item when $\Theta>1$ there is absence of calendar spread arbitrage if and only if condition (\ref{condizione_necessaria_2}) holds jointly with
\[\Phi\leq 1\quad\text{ or }\quad (\Theta\Phi\rho_{2}-\rho_{1})^{2}\leq(\Theta-1)(\Theta\Phi^{2}-1);\]
\item when $\Theta>1$ and condition (\ref{condizione_necessaria_2}) holds jointly with
\[\Phi\leq 1\quad\text{ or }\quad (\Theta\Phi\rho_{2}-\rho_{1})^{2}<(\Theta-1)(\Theta\Phi^{2}-1)\]
there are no intersection points (i.e. $w_{1}(k)< w_{2}(k)$ for all $k\in\mathbb{R}$)
\item when $\Theta>1$, $\Phi>1$ and $(\Theta\Phi\rho_{2}-\rho_{1})^{2}=(\Theta-1)(\Theta\Phi^{2}-1)$ the two slices have exactly one intersection point which is a tangency point;
\item when $\Theta>1$, $\Phi>1$ and $(\Theta-1)(\Theta\Phi^{2}-1)<(\Theta\Phi\rho_{2}-\rho_{1})^{2}<(\Theta\Phi-1)^{2}$ there must exist exactly two points where the slices $w_{1}(k)$ and $w_{2}(k)$ cross over each other.
\item when $\Theta>1$, $\Phi>1$ and $(\Theta\Phi\rho_{2}-\rho_{1})^{2}=(\Theta\Phi-1)^{2}$ there must exist exactly one point where the slices $w_{1}(k)$ and $w_{2}(k)$ cross over each other.

\end{itemize}

\end{proposition}

%
%
%
%
%
%

\section*{Figures}

\begin{figure}[h!]
\centering
\begin{subfigure}[b]{0.45\textwidth}
\includegraphics[width=\textwidth]{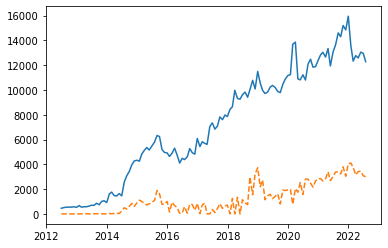}
\caption{Total number of available SPXW options before (continuous) and after (dashed) the filtering procedure.}
\end{subfigure}
\hfill
\begin{subfigure}[b]{0.45\textwidth}
\includegraphics[width=\textwidth]{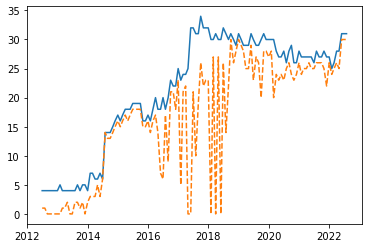}
\caption{Total number of available maturity dates before (continuous) and after (dashed) the filtering procedure.}
\end{subfigure}
\caption{The graphs show the total number of available options and available maturity dates in each option chain before and after filtering.}
\label{Figura_serie_storiche_opzioni_disponibili}
\end{figure}

\begin{figure}[h!]
\centering
\begin{subfigure}[b]{0.45\textwidth}
\centering
\includegraphics[width=\textwidth]{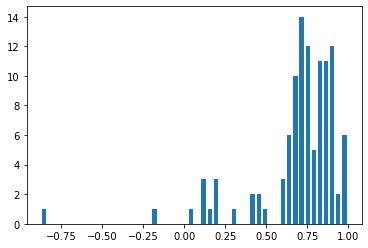}
\caption{Correlation coefficients vega versus bid-ask-spread.}
\label{istogramma_coeff_correlazione_vega_bidaskspread}
\end{subfigure}
\hfill
\begin{subfigure}[b]{0.45\textwidth}
\centering
\includegraphics[width=\textwidth]{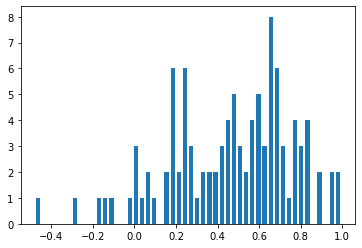}
\caption{Correlation coefficients DTE versus bid-ask-spread}
\label{istogramma_coeff_correlazione_DTE_bidaskspread}
\end{subfigure}
\caption{Histograms based on $108$ correlation coefficients (one coefficient for each option chain).} 
\label{correlations}
\end{figure}

\begin{figure}[h!]
\centering
\begin{subfigure}[b]{\textwidth}
\includegraphics[width=0.45\textwidth,height=!]{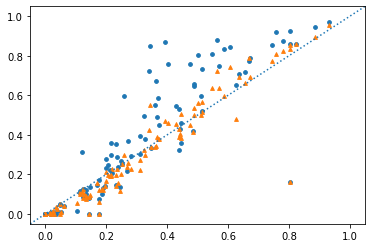}
\hfill
\includegraphics[width=0.45\textwidth,height=!]{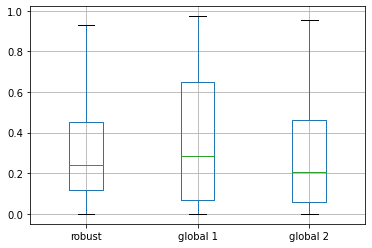}
\caption{$F_{1}$ measure of fit.}
\end{subfigure}

\bigskip

\begin{subfigure}[b]{\textwidth}
\includegraphics[width=0.45\textwidth,height=!]{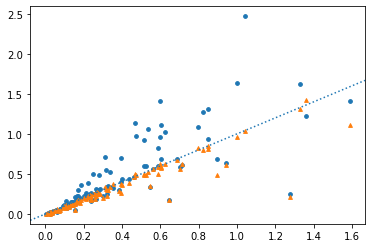}
\hfill
\includegraphics[width=0.45\textwidth,height=!]{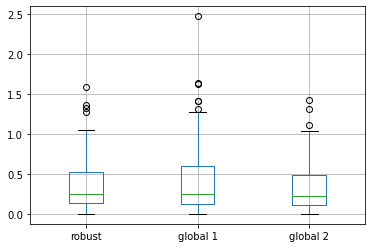}
\caption{$F_{2}$ measure of fit.}
\end{subfigure}

\bigskip

\begin{subfigure}[b]{\textwidth}
\includegraphics[width=0.45\textwidth,height=!]{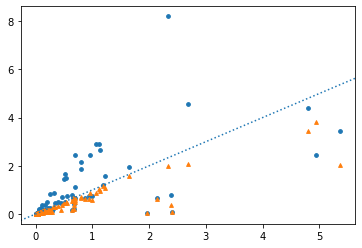}
\hfill
\includegraphics[width=0.45\textwidth,height=!]{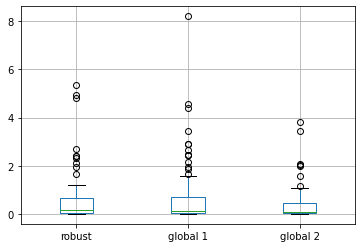}
\caption{$F_{3}$ measure of fit.}
\end{subfigure}
\end{figure}

\begin{figure}\ContinuedFloat
\begin{subfigure}[b]{\textwidth}
\includegraphics[width=0.45\textwidth,height=!]{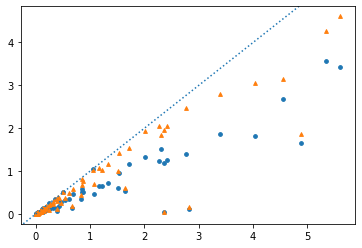}
\hfill
\includegraphics[width=0.45\textwidth,height=!]{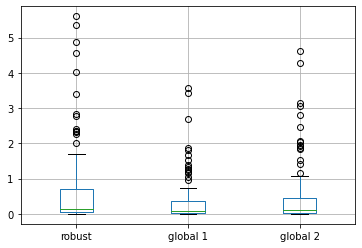}
\caption{$F_{4}$ measure of fit.}
\end{subfigure}
\caption{Scatterplots and boxplots for the measures of fit obtained on the $108$ calibrated eSSVI surfaces. \textit{"robust"} refers to the algorithm of \citet{Corbetta_Cohort_Laachir_Martini_2019} as described in Section \ref{algoritmo_robusto}. \textit{"global 1"} and \textit{"global 2"} refer to the algorithm of \citet{Mingone_2022} with inverse quadratic vega weights and with constant weights, respectively (see Section \ref{algoritmo_globale}). The abscissa values in the scatterplots are the $F_{i}$ values for the "robust" algorithm; the ordinate values are the $F_{i}$ values for the other two algorithms: circles refer to "global 1", triangles to "global 2".}
\label{Figura_measures_fit}
\end{figure}

\begin{figure}[h!]
\centering
\begin{subfigure}[b]{0.45\textwidth}
\includegraphics[width=\textwidth]{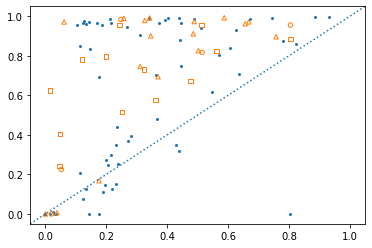}
\caption{Squared Inverse vega weights}
\end{subfigure}
\hfill
\begin{subfigure}[b]{0.45\textwidth}
\includegraphics[width=\textwidth]{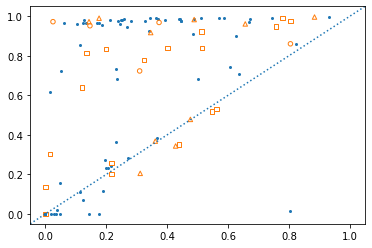}
\caption{Constant weights}
\end{subfigure}
\caption{Scatterplots of the $F_{1}$ measure of fit for the global algorithm. The $x$-axes refer to the implementation with initial values from robust algorithm, the $y$-axes to the implementation with the less data-driven initial values given by $\rho_{i}=0$, $a_{i}=\max\{\theta^{*}_{i}-\theta^{*}_{i-1}, 0\}$ and $c_{i}=0.5$ for $i=1,2,\dots, n$. Dots refer to cases where both choices for the initial values led to convergence in less than 500 objective function evaluations. Triangles refer to cases where the robust initial values led to convergence in less than 500 objective function evaluations, but the less data-driven initial values did not. Squares refer to cases where the robust initial values did not lead to convergence in less than 500 objective function evaluations, but the less data-driven initial values did. Circles refer to cases where none of the two choices for the initial values led to convergence in less than 500 objective function evaluations.}
\label{Figure_test_initial_values}
\end{figure}

\begin{figure}[h!]
\centering
\begin{subfigure}[b]{0.45\textwidth}
\centering
\includegraphics[width=\textwidth]{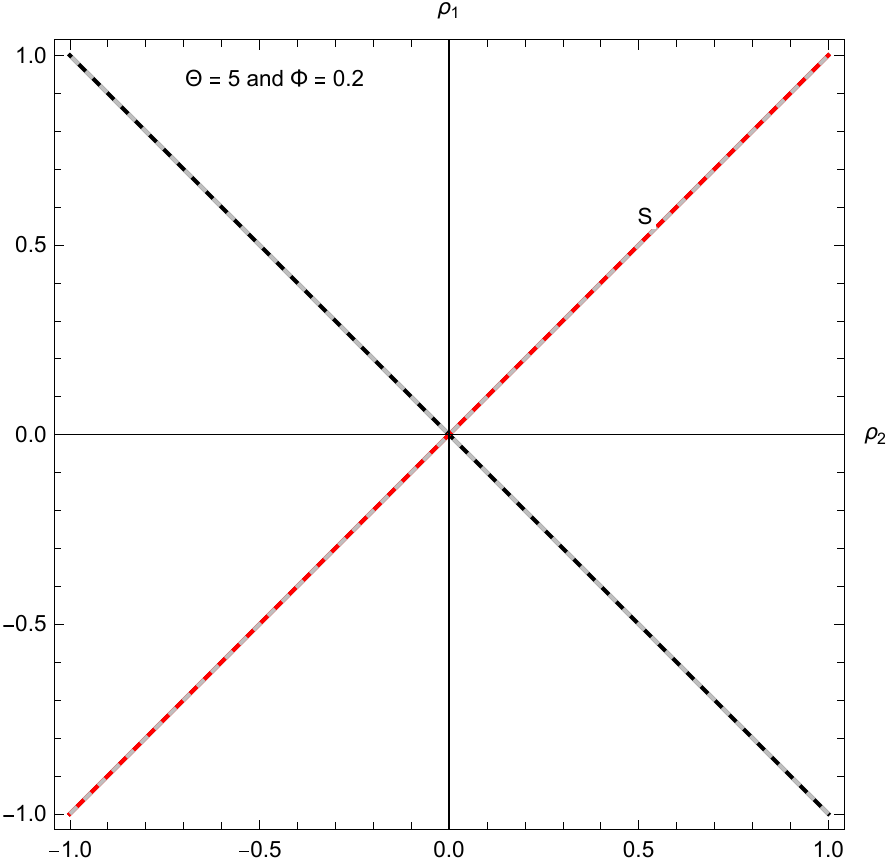}
\caption{$\Theta\Phi=1$ ($\Rightarrow$ $\Theta\Phi^{2}<1$)}
\label{sample-figure_part_a}
\end{subfigure}
\hfill
\begin{subfigure}[b]{0.45\textwidth}
\centering
\includegraphics[width=\textwidth]{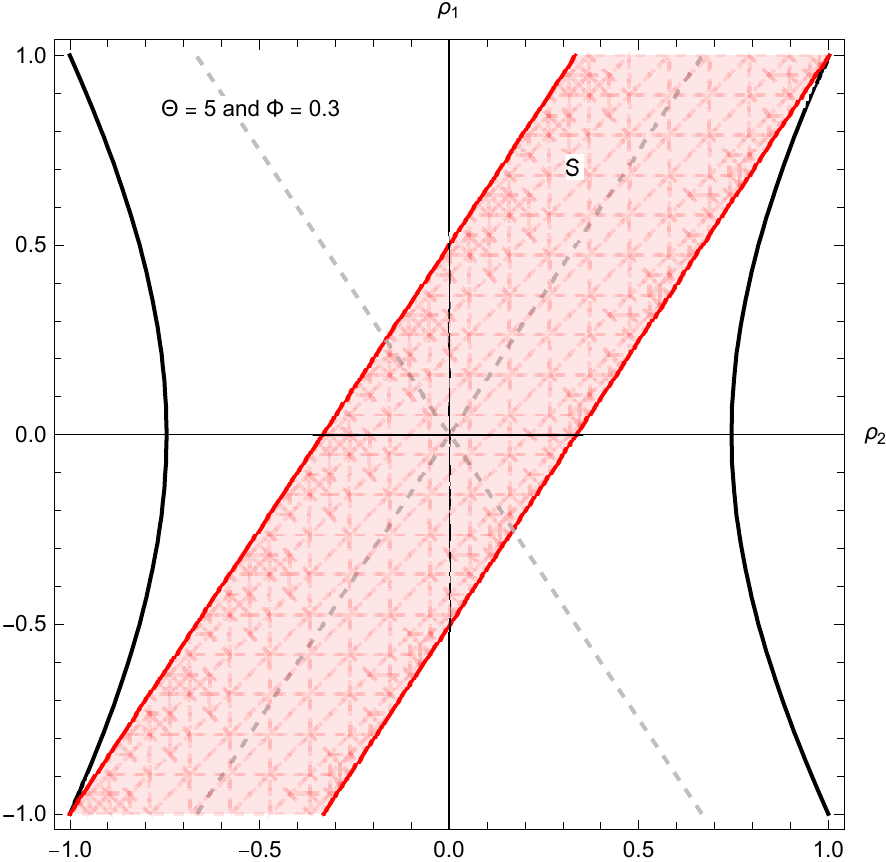}
\caption{$\Theta\Phi>1$ and $\Theta\Phi^{2}<1$}
\label{sample-figure_part_b}
\end{subfigure}\\

\bigskip

\begin{subfigure}[b]{0.45\textwidth}
\centering
\includegraphics[width=\textwidth]{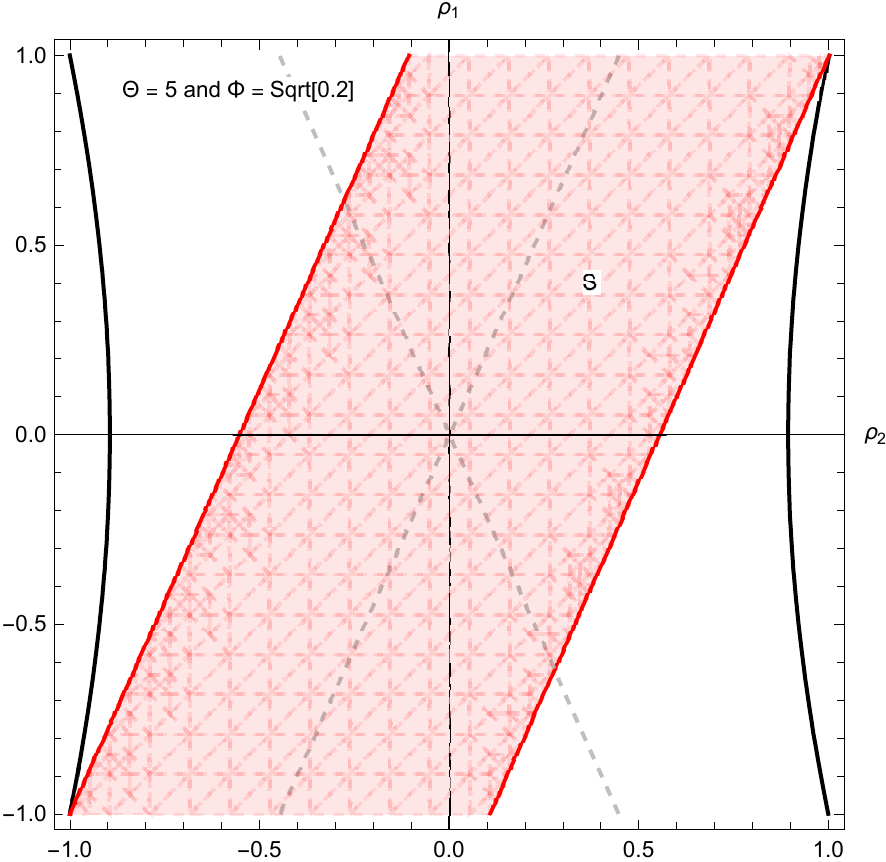}
\caption{$\Theta\Phi>1$ and $\Theta\Phi^{2}=1$}
\label{sample-figure_part_c}
\end{subfigure}
\hfill
\begin{subfigure}[b]{0.45\textwidth}
\centering
\includegraphics[width=\textwidth]{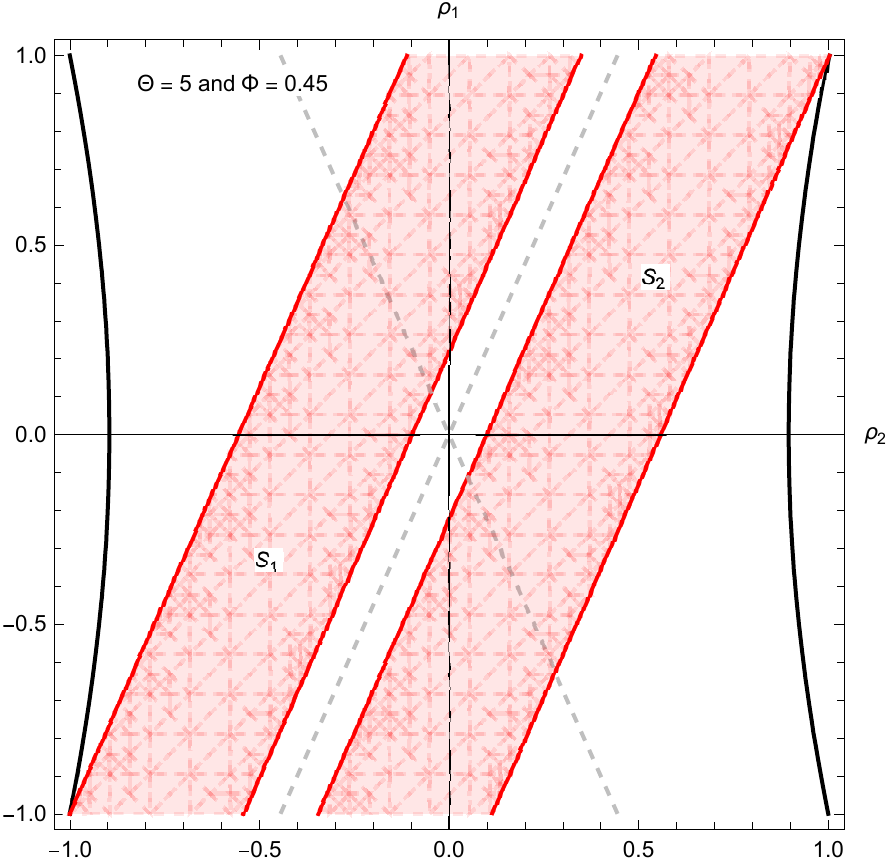}
\caption{$\Theta\Phi>1$ and $\Theta\Phi^{2}>1$}
\label{sample-figure_part_d}

\end{subfigure}\\

\bigskip

\begin{subfigure}[b]{0.45\textwidth}
\centering
\includegraphics[width=\textwidth]{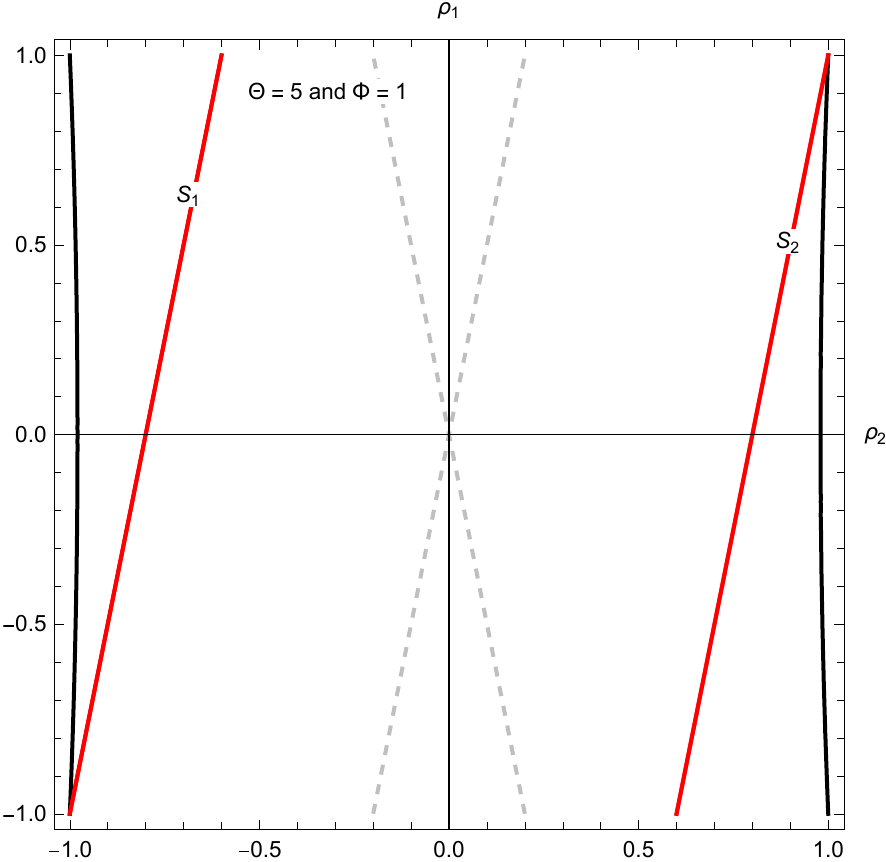}
\caption{$\Theta\Phi>1$ and $\Phi=1$}
\label{sample-figure_part_e}
\end{subfigure}\\

\caption{The red areas show all possible shapes of the $R_{\Theta, \Phi}$ when $\Theta>1$ and $\Theta\Phi\geq 1$. The black line is the graph of the hyperbola $H_{\Theta,\Phi}$. The two dashed lines are the asymptotes $\rho_{1}=\pm\Theta\Phi\rho_{2}$ of $H_{\Theta, \Phi}$.}
\label{Figura_insieme_ammissibile}
\end{figure}

\end{document}